\documentclass[12pt]{article}

\usepackage[english]{babel}
\usepackage[utf8]{inputenc}
\usepackage[authoryear,longnamesfirst]{natbib}
\usepackage[margin=1in]{geometry}
\usepackage{setspace}
\usepackage{subcaption}
\usepackage{titlesec}
\titlespacing*{\section}{0pt}{1.5ex plus 1ex minus .2ex}{1ex plus .2ex}

\usepackage{amsmath}
\usepackage{amsfonts}
\usepackage{amsthm}
\usepackage{amssymb}
\usepackage{mathrsfs}
\usepackage{amstext}
\usepackage{comment}
\usepackage{epsfig}
\usepackage{latexsym}
\usepackage{epstopdf}
\usepackage{bm}
\usepackage{algorithm}
\usepackage{algpseudocode}
\usepackage{subcaption}

\floatname{algorithm}{Procedure}
\usepackage{footnote}
\usepackage{microtype}
\usepackage[inkscapelatex=false]{svg}

\DeclareMathOperator*{\argmin}{argmin}

\usepackage{rotating}    
\usepackage{pdflscape}

\usepackage{accents}
\usepackage{url}
\usepackage[colorlinks = true,linkcolor = blue]{hyperref}
\hypersetup{
        urlcolor = blue,
	colorlinks = true,
	linkcolor = blue,
	citecolor = blue,
}
\usepackage{enumitem}

\usepackage{multirow}
\usepackage{multicol}
\usepackage{babel}
\usepackage{array}
\usepackage{rotating}
\usepackage{longtable}
\usepackage{float}
\usepackage{booktabs}
\usepackage{lscape}

\newtheorem{theorem}{Theorem}

\newtheorem{assumption}{Assumption}
\newtheorem{lemma}{Lemma}

\title{A Synthetic Business Cycle Approach to Counterfactual Analysis with Nonstationary Macroeconomic Data\thanks{Authors are listed in alphabetical order. We thank participants of the CUHK Econometrics Workshop for their valuable comments.}}
\author{Zhentao Shi\thanks{The Chinese University of Hong Kong. Email: \url{zhentao.shi@cuhk.edu.hk}. 
Shi acknowledges the partial financial support from the National Natural Science Foundation of China (Project No.~72425007).
} \and Jin Xi\thanks{Corresponding author. AMSS Center for Forecasting Sciences, Chinese Academy of Sciences.
Address: No.~55 Zhongguancun East Road, Academy of Mathematics and Systems Science, Chinese Academy of Sciences, Beijing 100190, China.
Email: \url{xijin@amss.ac.cn}.} \and Haitian Xie\thanks{Peking University. Email: \url{xht@gsm.pku.edu.cn}. Xie gratefully acknowledges the financial support from the National Natural Science Foundation of China (No.~72403008, No.~72495123).}}
\date{\today}

\begin{document}

\maketitle

\begin{abstract}
\singlespacing

This paper investigates the use of synthetic control methods for causal inference in macroeconomic settings when dealing with possibly nonstationary data. While the synthetic control approach has gained popularity for estimating counterfactual outcomes, we caution researchers against assuming a common nonstationary trend factor across units for macroeconomic outcomes, as doing so may result in misleading causal estimation—a pitfall we refer to as the spurious synthetic control problem. To address this issue, we propose a synthetic business cycle framework that explicitly separates trend and cyclical components. By leveraging the treated unit's historical data to forecast its trend and using control units only for cyclical fluctuations, our divide-and-conquer strategy eliminates spurious correlations and improves the robustness of counterfactual prediction in macroeconomic applications. As empirical illustrations, we examine the cases of German reunification and the handover of Hong Kong, demonstrating the advantages of the proposed approach.


\end{abstract}

\newpage

\section{Introduction}

Milestone socio-economic events, such as the German reunification and the handover of Hong Kong
(where one of the coauthors once studied and another currently resides), 
are defining moments of policy reforms, economic shocks, and geopolitical changes at macro-level units (e.g., countries, states, or regions).
Assessing the impacts of events of such scales can only be carried out by observational data, instead of randomized controlled trials. 

A prominent approach in this domain is the synthetic control method \citep{abadie2003economic, abadie2010synthetic}, which estimates the untreated potential outcome of the treated unit by constructing a weighted average of control units' outcomes. The weights are typically derived by minimizing the discrepancy between the treated unit's pre-treatment outcomes and those of the control units.\footnote{Throughout the paper, we use the terms ``control units'' and ``donor units'' interchangeably.} 

Although widely used, the synthetic control method has received limited formal attention regarding issues posed by nonstationary data.\footnote{See the literature review part for a discussion of the few existing contributions.} This gap is surprising, given that empirical studies frequently apply the synthetic control method to potentially nonstationary macroeconomic variables, including GDP \citep{abadie2003economic,billmeier2013assessing,abadie2015comparative}, gross state product \citep{ben2021augmented}, government spending \citep{eliason2018can,asatryan2018balanced}, exchange rate \citep{chamon2017fx}, unemployment \citep{gobillon2016regional,peri2019labor}, gasoline price \citep{becker2021price}, asset price \citep{acemoglu2016value}, carbon emission \citep{andersson2019carbon}.\footnote{These variables are often regarded as potentially nonstationary. See, for example, \cite{mccracken2016fred} for discussions on first-differencing and other transformations commonly used to mitigate the concerns of nonstationarity in these variables.
} 

This paper cautions empirical researchers who attempt to apply the synthetic approach to examine causal effects using potentially nonstationary macroeconomic datasets. 
Not only would the causal conclusions be invalid when nonstationarity is not pervasive across units, but researchers might also empirically observe that control units provide seemingly good estimates of the treated unit prior to the intervention, even if the control units are entirely unrelated to the treated unit and lack any predictive power after the treatment period. This phenomenon of spurious relationships was well-documented in the context of linear regression \citep{granger1974spurious,phillips1986understanding}, and later was extended to factor models \citep{onatski2021spurious}.

When analyzing nonstationary data, it is common economic practice to decompose a nonstationary time series into two distinct components: a trend component and a cyclical component 
\citep{beveridge1981new, hodrick1997postwar}.
 The trend component reflects long-term economic growth or decline, driven largely by the unit's intrinsic characteristics such as institutional quality, technological progress, labor force growth, and structural economic conditions. For example, the steady rise in GDP per capita in advanced economies over the past century can be attributed to factors like innovation, human capital accumulation, and institutional development, which shape the trend component. In contrast, the cyclical component captures short- to medium-term fluctuations associated with the business cycle, which are influenced by common factors such as global economic conditions, commodity price shocks, or financial market dynamics. 

By isolating components driven by fundamentally different forces, the trend-cycle decomposition provides a new perspective to address nonstationarity within the synthetic control framework.
The key premise of the synthetic control approach is that the untreated potential outcomes of both control and treated units share a common underlying structure that can be characterized by a factor model, as specified by \cite{abadie2010synthetic}:
\begin{align} \label{eqn:old factor characterization}
    Y_{i,t}(0) = \lambda_i f_t + \epsilon_{i,t},
\end{align}
where $\lambda_i$ and $f_t$ respectively represent individual factor loadings and the vector of common factors, and $\epsilon_{i,t}$ denotes idiosyncratic terms. This naturally leads to the question: Which component—--the trend, the cycle, or both---is compatible with the synthetic control framework?

Business cycles' comovement or synchronization across countries is a well-documented phenomenon in the economic literature, reflecting the financial and firm-level linkages of economies \citep{di2018micro,avila2024common}. 
For example, the recession of 2008-2009 led to synchronized downturns across many economies, illustrating the shared nature of their cyclical components, as specified in (\ref{eqn:old factor characterization}).
Nevertheless, imposing a common nonstationary trend across units can be problematic.
The reasons for this are twofold. First, as previously discussed, the longer-term dynamics of each country could be primarily determined by its intrinsic factors rather than shared features across countries. Second, when the nonstationarity of the treated unit  cannot be captured by the control units, imputing its potential outcomes using a weighted average of the control units’ outcomes can yield misleading results due to spurious regression. This is because the imputation is essentially based on linearly regressing the treated unit's potential outcomes on those of the control units. As a result, even if the two groups are entirely unrelated in their trends, researchers may still observe the empirically significant explanatory power before the treatment occurs. This issue of spurious synthetic control will be further discussed in Section \ref{sec:conventional}.

To address this challenge, we propose a new method that distinguishes between the trend and cyclical components within the synthetic control framework. Our approach uses the treated unit's own historical data to forecast its trend component, both because the high persistence of the trend makes it predictable from its own past values, and because this component is mostly driven by unit-specific characteristics that are not easily captured by control units.
Meanwhile, we utilize the cyclical components of the control units to construct synthetic predictions for the treated unit's cyclical behavior, as these components are likely to be influenced by common short-term factors shared across units. We refer to this new approach as the \emph{synthetic business cycle approach} (SBC). 

By modeling the trend and cyclical dynamics separately, our divide-and-conquer strategy provides a robust and valid framework for estimating counterfactual outcomes. It strengthens the validity of causal inferences by addressing the spurious synthetic control problem, offering a reliable tool for causally analyzing macroeconomic events. We formally show that our estimator provides an asymptotically unbiased estimate of the counterfactual outcome given the trend-cycle decomposition.
For counterfactual prediction in the nonstationary world, this is the first procedure, to the best of our knowledge, that is robust regardless of whether cointegration is present or not; it provides agnostic practitioners with an easy-to-implement algorithm supported by theoretical justifications.

We proceed with extensive numerical work to showcase the finite-sample performance of our method. 
Our simulation results show that when the correlation between the treated unit and the control units is spurious, the proposed synthetic business cycle estimator achieves substantial reductions in mean-squared error (MSE) for counterfactual prediction. When some donors are cointegrated with the treated unit---a scenario in which the conventional method should perform well---our estimator maintains strong performance. 

For our empirical applications, we revisit two classic applications of the synthetic control method: German reunification and Hong Kong’s 1997 return, both focusing on the nonstationary outcome GDP per capita. Although both methods fit the pre-treatment data well, their post-treatment counterfactuals diverge. The synthetic weights for the trend and cycle components select different donor groups, and the conventional synthetic weights lean heavily toward the trend, producing the discrepancy. Placebo tests confirm the robustness of the proposed estimator, whereas the conventional estimator diverges from the actual data.

\paragraph{Literature review.}
This study is based on multiple strands of literature. 
Synthetic control is arguably the most popular framework of counterfactual prediction with observational data, which has witnessed generalization and explorations in all directions, for example
\cite{ben2021augmented},
\cite{doudchenko2016balancing},
\cite{botosaru2019role},
and
\cite{carvalho2018arco},
to name a few. 
The key idea of synthetic control is to leverage the control units in the donor pool to predict the counterfactual after the event happens, and therefore it naturally involves the time dimension as well as the cross section, i.e., a panel data structure. \emph{Panel data approach for program evaluation}, initiated by \cite{hsiao2012panel}, and is followed by 
\cite{hou2021revisiting},
\cite{li2024confidence},
and 
\cite{shi2023forward}, etc, is a cousin of synthetic control. 
When the panel data is split into pre-treatment / post-treatment and treated / control quadrants, the counterfactual prediction exercise is associated with the statistical problem of \emph{matrix completion} 
\citep{athey2021matrix,bai2021matrix, xiong2023large}.
We will present our model in the matrix completion perspective with outcome variables alone; other covariates are abstracted from the context.

The only papers, as we are aware of, that theorize the counterfactual prediction in the nonstationary  environment are \cite{bai2014property}, \cite{masini2022counterfactual}
and \cite{masini2021counterfactual}.
The basic procedures of these three studies are similar: directly using a regression to project the target time series onto the nonstationary multivariate time series of the donor pool. 
The first and the second papers work under the low-dimensional environment and therefore adopt the OLS estimation, whereas the third one extends to a diverging number of regressors and thus regularizes the linear regression with sparsity. 
The direct regression approach works well in the presence of cointegration by finding the cointegration vector.
\cite{bai2014property} assumes nonstationary factors where spurious regression is ruled out, while
\citet[p.232]{masini2022counterfactual} acknowledges ``\textit{As a consequence, in this case (of cointegration) no method based on linear combination of peer will give a reliable counterfactual,}'' and 
\citet[p.1779]{masini2021counterfactual} suggests ``\textit{To avoid spurious results, the practitioner can pretest for cointegration.}''
Pre-testing errors, however,  will be carried into the procedure and thus distort the estimator.

Faced with the challenges imposed by potential spurious regressions, we tackle nonstationarity with a filtering procedure, and then the remaining stationary components are further recast into the synthetic control framework. 
In the language of machine learning, the filtering step carries out \emph{feature engineering} that pre-processes the nonstationary time series, which is difficult to handle via regressions at the risk of spurious regressions, whereas the weighted linear combination step is standard supervised learning.  
In terms of algorithm design, our method takes a distinctive strategy from the three papers referred to in the last paragraph, and therefore is new in the literature.

Recently, \citet{shen2023same} showed a surprising result: vertical regression (synthetic control–type methods), horizontal regression, and synthetic difference-in-differences\footnote{Horizontal regression first regresses each control unit’s post-treatment outcomes on its own pre-treatment history, then applies the estimated model to the treated unit to predict its counterfactual outcome \citep{athey2021matrix}. The synthetic difference-in-differences estimator introduced by \citet{arkhangelsky2021synthetic} integrates both vertical and horizontal regression components.} yield numerically identical counterfactual predictions, up to the effect of penalty terms. This suggests that these popular approaches for panel-data causal analysis handle the time and unit dimensions symmetrically. However, such symmetry is suboptimal in nonstationary settings, as it overlooks the strong temporal predictability inherent in these data. In contrast to those above procedures, our procedure is intentionally asymmetric in time and unit, leveraging the temporal structure and the persistence of nonstationarity to improve counterfactual prediction.


\paragraph{Organization.} 
The rest of the paper is organized as follows. 
Section \ref{sec:conventional} introduces the panel causal inference framework, outlines the conventional synthetic control procedure, and discusses the spurious synthetic control problem. Section \ref{sec:SBC} presents the synthetic business cycle approach and derives its theoretical properties. Section \ref{sec:simulation} reports the simulation studies. Section \ref{sec:empirical} applies the proposed method to two empirical settings, German reunification and the return of Hong Kong. Section \ref{sec:conclusion} concludes.

\section{Synthetic Control Method} \label{sec:conventional}

\subsection{Setup and Procedure}

The standard synthetic control works in a dataset of a panel structure.
Consider a panel of $N+1$ units, $i=1,\dots,N+1$, observed over $t=1,\dots,T$. Without loss of generality, let the first unit $i=1$ be the treated unit, receiving treatment from period $T_0+1$ onward for some $1 < T_0 \leq T$, while the other units $i=2,\dots,N+1$ remain untreated throughout. Although this setup focuses on a single treated unit, the framework extends directly to multiple treated units by applying the method separately to each treated unit \citep{abadie2021using}.

We are interested in identifying and estimating the effect of the treatment on some outcome variable $Y$ of the treated unit. In the potential outcome notation, the treated unit has a pair of potential outcomes $Y_{1,t}(1)$ and $Y_{1,t}(0)$ at each period $t$. We observe $Y_{1,t} = Y_{1,t}(0)$ for pre-treatment periods $t \leq T_0$, whereas $Y_{1,t} = Y_{1,t}(1)$ for post-treatment periods $t > T_0$. For the untreated units, the potential outcome $Y_{i,t}(0)$ is observed for all time periods. A treated potential $Y_{i,t}(1)$ need not be defined for these control units, as in many applications it is not meaningful to conceptualize what their outcomes would have been under the treatment.  

The treatment effect on the treated unit at a post-treatment period $t$ is:
\begin{align*}
    Y_{1,t}(1) - Y_{1,t}(0), \quad t > T_0.
\end{align*}
Since $Y_{1,t}(1)$ is observed for all $t > T_0$, causal inference is reduced to imputing the unobserved counterfactual $Y_{1,t}(0)$ in the post-treatment periods. Equivalently, one can view the panel-data causal inference problem as a matrix-completion task \citep{athey2021matrix}. In particular, letting the rows index units $i=1,\dots,N+1$ and the columns index periods $t=1,\dots,T$, the matrix of untreated potential outcomes,
\[
\begin{pmatrix}
\checkmark & \checkmark & ?        & \dots & ?\\
\checkmark & \checkmark & \checkmark & \dots & \checkmark\\
\checkmark & \checkmark & \checkmark & \dots & \checkmark\\
\vdots     & \vdots     & \vdots     & \ddots& \vdots\\
\checkmark & \checkmark & \checkmark & \dots & \checkmark
\end{pmatrix}%
\smash{\raisebox{8ex}{\;\text{← treated unit}}}
\]
is fully observed except for the entries corresponding to $i=1$ and $t > T_0$, which must be imputed.

The standard synthetic control approach method designed by \cite{abadie2003economic} and the subsequent papers imputes the unobserved post-treatment $Y_{1,t}$ by weighting the control units' outcomes as
\begin{align*}
\tilde{Y}_{1,t}(0) \equiv \sum_{i=2}^{N+1} \tilde{w}_{i} Y_{i,t}, 
\quad t \geq T_0 +1,
\end{align*}
where the weights are obtained by fitting the pre-treatment outcomes 
\begin{align*}
 (\tilde{w}_2,\cdots,\tilde{w}_{N+1}) \equiv \argmin_{(w_2,\cdots,w_{N+1})} \sum_{t \leq T_0} \left( Y_{1,t} - \sum_{i=2}^{N+1} w_{i} Y_{i,t} \right)^2 \text{ subject to } \tilde{w}_i \geq 0, \sum_{i=2}^{N+1} w_{i}=1.
\end{align*}
More recent literature frames the synthetic control within the broader vertical regression paradigm\citep{hsiao2012panel,doudchenko2016balancing,athey2021matrix}. Specifically, the vertical regression method imputes the missing outcome $Y_{1,t}(0)$ in the post-treatment periods as
\begin{align*}
    \hat{Y}_{1,t}(0) \equiv \hat{\gamma}_1 + \sum_{i=2}^{N+1} \hat{\gamma}_{i} Y_{i,t}, \quad  t \geq T_0 +1,
\end{align*}
where
\begin{align*}
    (\hat{\gamma}_2,\cdots,\hat{\gamma}_{N+1}) \equiv \argmin_{(\gamma_2,\cdots,\gamma_{N+1})} \sum_{t \leq T_0} \left( Y_{1,t} - \gamma_1 - \sum_{i=2}^{N+1} \gamma_{i} Y_{i,t} \right)^2.
\end{align*}
The vertical regression allows for a non-zero intercept and negative weights.

The intuition behind the synthetic control estimator is the principle that ``similar units behave similarly'' \citep{shen2023same}. For example, \cite{abadie2010synthetic} and many subsequent studies motivate the method using a common factor structure for the untreated potential outcomes. When all units are influenced by the same latent factors, a suitably weighted combination of untreated units can approximate the trajectory of the treated unit. This principle makes synthetic control an intuitively appealing and transparent approach to causal inference in settings with long panel data.

\subsection{Nonstationarity and the Spurious Synthetic Control Problem}


When the outcomes are nonstationary, the ``similar units behave similarly'' justification for synthetic control is under threat due to the well-known problem of spurious regression \citep{granger1974spurious,phillips1986understanding}. Concretely, if each untreated potential outcome $Y_{i,t}(0)$ follows an independent nonstationary (e.g., unit-root) process, then using control units’ outcomes to impute the treated unit’s counterfactual is invalid. Yet, a vertical regression of $Y_{1,t}$ on the control outcomes over the pre-treatment period will often produce statistically significant coefficients, an artifact of spurious comovement rather than genuine shared factors.

Imposing nonnegativity and unit-sum constraints on the weights, as in standard synthetic control, does not resolve the issue. Those constraints narrow the feasible set of weights but do not eliminate the ability of independent nonstationary series to falsely explain one another. Spurious comovement, therefore, remains a serious concern whenever synthetic control is applied to nonstationary outcomes. These several issues will be demonstrated in the simulation results reported in Section~\ref{sec:simulation}.

As a practical illustration, consider researchers analyzing the GDP trajectories of multiple countries. Even if a country’s GDP can be closely approximated by a weighted average of others over a given period, such a fit may arise purely from coincidental trending behavior. There is no reason to expect this relationship to persist out of sample, and, therefore, it cannot be used to impute the treated country’s unobserved potential GDP. This is the spurious synthetic control problem.

Another concern is that, although the factor model heuristic motivates the synthetic control method, the procedure does not explicitly estimate the underlying latent factors. As a result,
researchers often neglect to examine whether the idiosyncratic term $\varepsilon_{it}$ in (\ref{eqn:c-factor}) is stationary or not, posing the risk of spurious synthetic control.
In response to this challenge, we propose a method that is robust to the presence of spurious relationships in the following section,

\section{Synthetic Business Cycle} \label{sec:SBC}

\subsection{Procedure}

We decompose the untreated potential outcome into two parts: a trend component $\tau_{i,t}$ and a cyclical component $c_{i,t}$:
\begin{align*}
    Y_{i,t}(0) & = \tau_{i,t} + c_{i,t}.
\end{align*}
The trend component captures the nonstationary part of the potential outcome, reflecting each unit's long-term evolution. The cyclical component is stationary and captures short‑term fluctuations that dissipate over time, such as recessions and expansions. Such a trend–cycle decomposition is a classic tool in macroeconomic analysis \citep{stock1988variable,stock1999business}.

Building on the trend–cycle decomposition of the potential outcome, we propose a synthetic business cycle approach for imputing $Y_{1,t}(0)$ in the post-treatment periods. 
This method follows a ``divide-and-conquer'' principle, as explained below.
\begin{itemize}[label={}, leftmargin=*, align=left]
    \item[\textbf{Step 1.}] In the pre-treatment period, apply a filter to decompose each unit's outcome series into its trend and cyclical components.
    \item[\textbf{Step 2.}] Use the treated unit's pre-treatment outcomes to forecast its post-treatment trend.
    \item[\textbf{Step 3.}] Use the cyclical components of the control units to construct a synthetic cyclical component for the treated unit.
    \item[\textbf{Step 4.}] Combine the predicted trend and synthetic cyclical component to obtain \(\hat{Y}_{1,t}(0)\) in the post-treatment periods.
\end{itemize}

The detailed procedure is operationalized below. In the first step, we apply the \cite{hamilton2018you} filter to detrend each series and obtain $\hat{\tau}_{i,t}$ and $\hat{c}_{i,t}$ for $t \leq T_0$ and $ 1 \leq i \leq N+1$. We discuss the rationale of choosing the Hamilton filter and compare it with alternative filtering methods in Section \ref{sec:filters}. \cite{hamilton2018you} proposed defining the cyclical component of a variable as the error made when forecasting its value at date $t$ using a linear function of $p$ of its own values observed up to date $t-h$, where $h$ is the forecasting horizon, typically recommended to be two to four years. The forecast error has a natural interpretation: it reflects the component that cannot be anticipated $h$ periods in advance, typically due to cyclical factors such as recessions.

Mathematically, Hamilton's approach defines the trend and cyclical components as
\begin{align} \label{eqn:tau-projection}
    \tau_{i,t} &\equiv {\mathbb{P}}(Y_{i,t}(0)|1,Y_{i,t-h}(0),Y_{i,t-h-1}(0),\dots
,Y_{i,t-h-p+1}(0)) \nonumber \\ 
& \equiv \alpha_{i,0} + \alpha_{i,1}Y_{i,t-h}(0)+\alpha_{i2}Y_{i,t-h-1}(0)+\dots
+\alpha_{i,p}Y_{i,t-h-p+1}(0), \\  
{c}_{i,t}& \equiv Y_{i,t}(0) - \tau_{i,t}, \nonumber
\end{align}
where $\mathbb{P}(y|x)$ denotes the population linear projection of $y$ on $x$, $(\alpha_{i,0},\cdots,\alpha_{i,p})$ are the corresponding projection coefficients, $h$ is the forecasting horizon, and $p$ is the number of self-lags used for prediction. The estimates $\hat{\tau}_{i,t}$ and $\hat{c}_{i,t}$ are obtained from the sample analogue of (\ref{eqn:tau-projection}). Specifically, we linearly project the pre-treatment $Y_{i,t}$ on a constant and its own lags to obtain the coefficients $(\hat{\alpha}_{i,0}, \dots, \hat{\alpha}_{i,p})$, which yield the fitted value $\hat{\tau}_{i,t}$ and hence the residual $\hat{c}_{i,t}$.

In the second step, we use the treated unit's lags prior to $T_0$ to form predictions of its post-treatment trend $(\hat{\tau}_{1,T_0+1}, \hat{\tau}_{1,T_0+2},\cdots, \hat{\tau}_{1,T_0+h})$ had the treatment never happened:
\begin{align*}
    \hat{\tau}_{1,t} \equiv \hat{\alpha}_{1,1} Y_{1,t-h} + \cdots + \hat{\alpha}_{1,p} Y_{1,t-h-p+1},\quad  T_0 +1 \leq  t \leq T_0+h.
\end{align*}

In the third step, we use the synthetic control method to impute the cyclical component of the treated unit using the control units: 
\begin{align*}
    \hat{c}_{1,t} \equiv \sum_{i=2}^{N+1} \hat{w}_{i} \hat{c}_{i,t}, \quad  t \geq T_0 + 1,
\end{align*}
where the weights are obtained by matching the pre-treatment cycles:
\begin{align} \label{eqn:cycle-weights}
 (\hat{w}_2,\cdots,\hat{w}_{N+1}) \equiv \argmin_{(w_2,\cdots,w_{N+1})} \sum_{t \leq T_0} \left( \hat{c}_{1,t} - \sum_{i=2}^{N+1} w_{i} \hat{c}_{i,t} \right)^2.
\end{align}
Note that an intercept term is not needed here, as the cyclical components are mean zero by construction.

\begin{figure}[t]
    \centering
    \includegraphics[width = 0.9\textwidth]{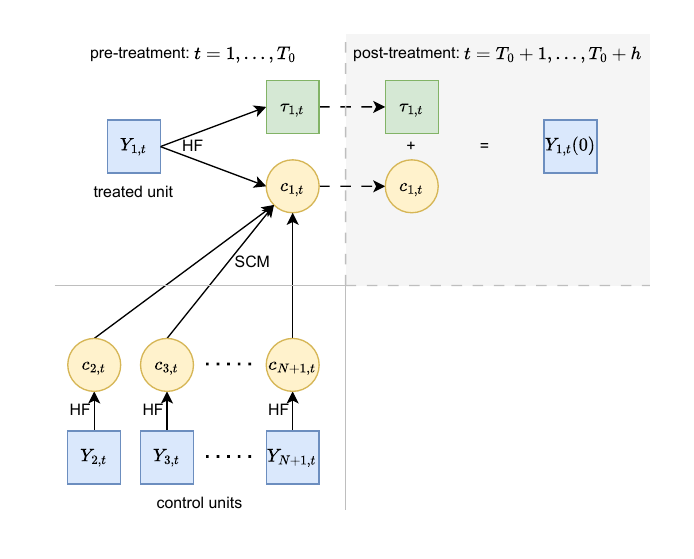}
    \caption{Diagram of the procedure}
    \label{fig:diag}
    \caption*{\footnotesize 
    The inputs of the learning process is $Y_{i,t}$ for $i=1,2,\ldots,N+1$ and $t = 1, \ldots, T_0$. The output is an estimator of the counterfactual $(Y_{1,t}(0))_{t=T_0+1}^{T_0+h}$.
    Objects in squares are nonstationary, whereas those in circles are stationary. The pre-treatment / post-treatment and treated / control combinations partition the plane into four quadrants, with the top-right one for counterfactual. The Hamilton filter (HF) is applied to each univariate time series separately for the trend-cycle decomposition.
    Each trend is kept standalone to be immune from spurious regression,
    while $c_{1,t+h}$ is supervised-learned by the cycles of the control units.}
\end{figure}

Lastly, we combine the trend and cyclical components to construct the imputed untreated potential outcome:
\begin{align*}
    \hat{Y}_{1,t}(0) \equiv \hat{\tau}_{1,t} + \hat{c}_{1,t},\  T_0 + 1 \leq t \leq T_0 + h.
\end{align*}
The estimated causal effect is then given by $Y_{1,t} - \hat{Y}_{1,t}(0)$. Figure \ref{fig:diag} presents a flowchart of the synthetic business cycle procedure.

The synthetic business cycle estimator offers several key benefits. First, by detrending each series to isolate a stationary cyclical component, it directly tackles spurious comovement in nonstationary data, as our theoretical results demonstrate. Second, it integrates seamlessly with the common-factor framework: even when idiosyncratic trends follow a factor structure, our detrending yields asymptotically unbiased estimates with small efficiency loss given a sufficiently long pre-treatment period. Third, it builds naturally on the conventional synthetic control method, preserving its intuitive weighting scheme and interpretability, while enhancing robustness in macroeconomic applications.

\subsection{Theoretical Results}

Following the literature \citep[e.g.,][]{
hsiao2012panel, masini2022counterfactual}, we adopt a ``fixed-$n$, large-$T$'' asymptotic framework. Below, we make formal assumptions about the common factor structure of the cyclical components and high-level assumptions regarding the estimation accuracy of the detrending filter.

\begin{assumption} \label{asm:c-factor}
    For each unit $i$, the cyclical component $c_{i,t}$ is weakly stationary and admits the following factor structure:
    \begin{align}
    c_{i,t} & = \lambda_i f_t + \varepsilon_{i,t}, \label{eqn:c-factor}
\end{align}
where $f_t \equiv (f_{1,t},\cdots,f_{L,t})'$ is a $L \times 1$ vector of stationary factors, $\lambda_i \equiv (\lambda_{i,1},\cdots,\lambda_{i,L})$ represents the corresponding deterministic loadings, and $\varepsilon_{i,t}$ is the idiosyncratic error that is mean zero with finite second moment, and is independent over $i$ and $t$. The factors satisfy the following conditions: 
\begin{enumerate}[label=(\roman*)]
        \item There exists a constant $\overline{F}$ such that $|f_{l,t}| \leq \overline{F}$ for all $l = 1,\cdots L$ and $t = 1,\cdots,T_0$. 
        \item Let $\xi(M)$ be the smallest eigenvalue of $\frac{1}{M} \sum_{t = T_0 - M +1}^{T_0} f_t'f_t$. $\xi(M)$ is bounded away from zero for any positive integer $M$: $\underset{M \in \mathbb{Z}^{+}}{\inf} \xi(M) \geq \underline{\xi} >0$. 
        \item $\frac{1}{\tilde{T}_0} \sum_{t = h+p}^{T_0} f_t'f_t \overset{p}{\rightarrow} \Sigma_{F_0}$, where $\tilde{T}_0 = T_0 - h_p+1$ is the effective pre-treatment sample size. 
    \end{enumerate}
\end{assumption}



Assumption \ref{asm:c-factor} posits that the cyclical components across units are driven by a common factor structure and idiosyncratic error terms. Parts (i)-(ii) of this assumption are also used in \cite{abadie2010synthetic}. Specifically, Assumption \ref{asm:c-factor}(i)
requires the factor to be uniformly bounded by some constant $\overline{F}$, while Assumption \ref{asm:c-factor}(ii) implies that $\frac{1}{M} \sum_{t = T_0 - M +1}^{T_0} f_t'f_t$ and $\Sigma_{F_0}$ in part (iii) are positive definite for any positive integer $M$.


Next, we impose high-level assumptions on the estimation error from the decomposition step. 


\begin{assumption} \label{asm:error convergence}
Let $\hat{u}_{i,t} \equiv \hat{c}_{i,t} - c_{i,t} = \tau_{i,t} - \hat{\tau}_{i,t}$. For $i = 1,2,\cdots, N+1$,
\begin{enumerate}[label=(\roman*)]
\item $\hat{u}_{i,t} = o_p(1)$ for each $t = h+p, \dots, T_0$ with $T_0\rightarrow \infty$.
\item $\sum_{t = h+p}^{T_0} \hat{u}_{i,t}^2 = o_p(T_0)$ with $T_0\rightarrow \infty$.
\end{enumerate}
\end{assumption}

It is worth noting that Assumption \ref{asm:error convergence} is not tied to any specific filtering method, including the Hamilton filter. The theoretical guarantees established under this assumption apply broadly to any detrending technique that satisfies the stated condition on estimation error. For illustration, we demonstrate that Assumption \ref{asm:error convergence} holds when unit roots and deterministic trends are removed using the Hamilton filter. 

The following theorem demonstrates the asymptotic unbiasedness of the synthetic business cycle estimator.
\begin{theorem}\label{thm:unbiasedness}
\ \begin{enumerate}[label=(\roman*)]
     \item Under the specification of trend and business cycle in (\ref{eqn:tau-projection}) and Assumptions \ref{asm:c-factor} - \ref{asm:error convergence},
the estimator $\hat{Y}_{1,t}(0)$ has the following asymptotic characterization as $T_0 \rightarrow \infty$:
\begin{align*}
\hat{Y}_{1,t}(0) - Y_{1,t}(0) = \sum_{i=2}^{N+1}\hat{w}_i (\varepsilon_{i,t} - \varepsilon_{1,t})+o_p(1),
\end{align*}
for the post-treatment periods $T_0+1 \leq t \leq T_0 + h$. In particular, $\hat{Y}_{1,t}(0)$ is an asymptotically unbiased estimator of $Y_{1,t}(0)$ at $T_0+1 \leq t \leq T_0 + h$. 
     \item Under the special case where the cyclical component follows an exact low-dimensional factor structure   $c_{i,t}  = \lambda_i' f_t$ without idiosyncratic shocks,
the synthetic business cycle estimator of $\hat{Y}_{1,t}(0)$ is consistent for each post-treatment period:
    \begin{align*}
        \hat{Y}_{1,t}(0) \overset{p}{\rightarrow} Y_{1,t}(0),\ \    T_0+1 \leq t \leq T_0 + h.
    \end{align*}
    \end{enumerate}

\end{theorem}

The above theorem shows that our method delivers an asymptotically unbiased estimator of $Y_{1,t}(0)$ for each post-treatment $t$. However, consistent estimation of $Y_{1,t}(0)$ at each post-treatment date is generally not attainable when the idiosyncratic error is present.\footnote{An alternative approach in the literature for establishing consistency of synthetic control is to define the estimand as the average treatment effect over the post-treatment periods. With a diverging number of post-treatment observations, this average effect can be consistently estimated. However, this strategy does not guarantee consistency of the estimated treatment effect at any specific post-treatment time point.} 
This is because the idiosyncratic part of the treated unit(s) cannot be estimated using the control units, as they lack a common structure. 
Nevertheless, as we demonstrate in part (ii), the signal part of the counterfactual remains consistently estimable, yielding an asymptotically unbiased estimate of the treatment effect. This result parallels the theoretical findings of \citet{abadie2010synthetic}.
In a special case where the cyclical components are only driven by common factors without idiosyncratic errors, consistency of $\hat{Y}_{1,t}$ for each post-treatment period can be established.

\subsection{Discussion of Filters} \label{sec:filters}

There are several well-established methods for decomposing trend and cyclical components. However, we find the Hamilton filter offers distinct advantages in this context. As the preceding analysis indicates, any detrending method used in the synthetic business cycle approach should meet three key criteria.

First, the cyclical component extracted by the filter should represent a stationary process; otherwise, the resulting estimates remain vulnerable to the spurious synthetic control problem. The following lemma establishes the stationarity of the cyclical component extracted with the Hamilton filter for a broad class of data-generating processes, including: (a) integrated processes such as random walks, and (b) processes that are stationary around deterministic time trends.

\begin{lemma} \label{lm:c-stationary}
$c_{i,t}$ defined in equation (\ref{eqn:tau-projection})
is stationary if 
  (1) $Y_{i,t}(0)$ is
stationary around a polynomial deterministic trend of time with order of $d_{i}\leq p$, and satisfies 
\begin{equation}\label{eqn:poly trend}
T^{-\frac{1}{2}}\sum_{t=1}^{[Tr]}(Y_{i,t}(0)-\delta _{i,0}-\delta _{i,1}t-\delta
_{i2}t^{2}\dots -\delta _{i,d_{i}}t^{d_{i}})\overset{d}{\rightarrow }\omega
_{i}B_{i}(r),
\end{equation}%
where $[Tr]$ is the largest integer no greater than $Tr$, $B_{i}$ is a standard
Brownian motion; or alternatively (2) if $d_{i}$-order differences of $Y_{i,t}(0)$ are
stationary for some $d_{i}\leq p$, and satisfy 
\begin{equation}\label{eqn: integrated}
T^{-\frac{1}{2}}\sum_{t=1}^{[Tr]}(\Delta ^{d_{i}}Y_{i,t}(0)-\mu _{i})\overset{d}{%
\rightarrow }\omega _{i}B_{i}(r),
\end{equation}%
where $\mu _{i}$ represents the population mean of $\Delta ^{d_{i}}Y_{i,t}(0)$.
\end{lemma}


Second, the filter must provide accurate estimates of both the trend and cyclical components, in accordance with Assumption \ref{asm:error convergence}. The Hamilton filter’s simple parametric structure facilitates compliance with this assumption. The next lemma shows that the Hamilton filter readily satisfies Assumption \ref{asm:error convergence}(i), the consistency of trend and cycle estimates.

\begin{lemma} \label{lm:hamilton-filter-consistency}
For each $i$, assume that the cyclical component $c_{i,t}$ is weakly stationary and that the covariance matrix of $(Y_{i,t-h},\cdots,Y_{i,t-h-p+1})$ exists and is non-singular.
$\tau_{i,t}$ and $c_{i,t}$ defined in equation (\ref{eqn:tau-projection}) can be consistently estimated by regressing $Y_{i,t}(0)$ on its past values $Y_{i,t-h}(0), Y_{i,t-h-1}(0), \cdots, Y_{i,t-h-p+1}(0)$ for $h+p \leq t \leq T_0 + h$. 
\end{lemma}

We next illustrate Assumption \ref{asm:error convergence}(i)-(ii) for the Hamilton filter using a
leading example. Consider a unit root process $Y_{i,t}(0) = Y_{i,t-1}(0) + \varepsilon_{i,t}$, where $\{\varepsilon_{i,t}\}$ is a martingale difference sequence. This implies that the cyclical component is $c_{i,t} = \sum_{s=0}^{h-1} \varepsilon_{i,t-s}$, which represents the accumulation of shocks within the $h$-period window. 
Linearly projecting $Y_{i,t}(0)$ on its past values, as in equation (\ref{eqn:tau-projection}), yields consistent estimators of the projection coefficients: $\hat{\alpha}_{i,1} \overset{p}{\rightarrow} 1$ and $\hat{\alpha}_{ij} \overset{p}{\rightarrow} 0$ for $2 \leq j \leq p$. Moreover, we can show that the convergence of the cyclical component is uniform over time. For illustration, suppose we take $h = 1$ and $p = 1$. The coefficient estimator in equation (\ref{eqn:tau-projection}) is 
\begin{align*}
    T(\hat{\alpha}_{i,1} -1 ) &= \frac{T^{-1}\sum Y_{i,t-1}(0)\varepsilon_{i,t}}{T^{-2}\sum( Y_{i,t-1}(0))^2}  \overset{d}{\rightarrow } \frac{\int_0^1 B_i(r)dB_i(r)}{\int_0^1[B_{i}(r)]^2 dr},
\end{align*}
where $B_i(r)$ represents standard Brownian motion. 
The estimation error of the cyclical component is therefore 
\begin{align*}
    \hat{u}_{i,t} = (1- \hat{\alpha}_{i,1})Y_{i,\lfloor Tr \rfloor}(0) &= O_p(T^{-1/2}) \frac{1}{\sqrt{T}} Y_{i,\lfloor Tr \rfloor}(0) \\
    &= O_p(T^{-1/2})O_p\left(\sqrt{r/T}\right) = o_p(1),
\end{align*}
and the convergence holds for all $r\in [0,1]$. We can further verify Assumption \ref{asm:error convergence} with the following convergence result: 
\begin{equation*}
\sum_{t=1}^{T}\hat{u}_{i,t}^{2}=\left[ T(\hat{\alpha}_{i}-\alpha _{i})%
\right] ^{2} \cdot T^{-2}\sum_{t=1}^{T}Y_{i,t-1}(0)^{2} 
\overset{d}{\rightarrow }\frac{\left[ \int_{0}^{1}B_{i}(r)]dB_{i}(r)\right]
^{2}}{\sigma_i^2 \int_{0}^{1}[B_{i}(r)]^{2}dr }. 
\end{equation*}%

Third, the filter must permit post-treatment extrapolation of the trend component. Because the synthetic business cycle procedure extrapolates the treated unit’s trend based on the pre-treatment observations, the filter should be one-sided, relying exclusively on past data to estimate trend and cycle. Symmetric filters that incorporate future observations are therefore unsuitable in this setting.

That said, the synthetic business cycle estimator represents just one particular application of filtering methods. In general, the choice of filter is tailored to the specific context. For discussions of the circumstances in which alternative filters, such as the Hodrick–Prescott filter, may be advantageous, see \citet{phillips2021boosting} and \citet{mei2024boosted}.

\subsection{A Machine Learning Perspective}

While the asymptotic unbiasedness in Theorem \ref{thm:unbiasedness} is rigorously established under the assumptions, a modern perspective of machine learning calls for a procedure to be versatile enough to accommodate deviations from tightly specified conditions. Here we provide a heuristic discussion of why our algorithm provides reasonable performance under a more general data-generating process.

Stationary multivariate time series models have been extensively studied over the years. Recently, \cite{miao2023high} work with a vector autoregression (VAR) of order $p$  with latent factors:
\begin{equation}\label{var}
    Y_t = \sum_{j=1}^p B_j Y_{t-j} + \tilde \Lambda \tilde f_t + \tilde \epsilon_t, 
\end{equation}
where $Y_t = (Y_{1,t},\ldots,Y_{N+1,t})'$ in our context is an $(N+1)$-vector, $\tilde{\epsilon}_t$ is the associated idiosyncratic error, $B_j$ is an $(N+1)\times (N+1)$ coefficient matrix, and $\tilde \Lambda$ is the $(N+1)\times L $ loading matrix.
The VAR system can be pushed forward for $h$ periods and then formulated by iterative substitution as:
\begin{equation}\label{var_h}
    Y_{t+h} = \sum_{j=1}^p B_j^{(h)} Y_{t-j} + \tilde 
 \Lambda^{(h)} \tilde  f_t + \tilde  \epsilon_t^{(h)}.
\end{equation}
A restricted VAR, as in \cite{mayoral2013heterogeneous}, assumes 
\begin{equation}\label{eq:b_diag}
B_j^{(h)} = \mathrm{diag}(b^{(h)}_{1,j},\ldots,b^{(h)}_{N+1,j} )
\end{equation}
to be diagonal. 
A further special case of the VAR model is the interactive fixed effects model \citep{bai2009panel,
moon2017dynamic} in panel data analysis, where the slope coefficients are reduced to  homogeneous scalar coefficients
$B_j^{(h)} = b_j^{(h)} I_{N+1}$, with the identity matrix $I_{N+1}$ of the compatible size.
While standard the panel data regression attempts to fit the entire multivariate time series $Y_{t+h}$, \cite{hsiao2022transformed}'s transformed estimator allows for choosing the weight of a specific cross-section unit, without loss of generality $i=1$, for normalization. 

In contrast to all the papers mentioned in the above paragraph, we consider \eqref{var_h} with nonstationary $Y$. For simplicity, assume every time series $(Y_{i,t})$ is nonstationary, and we denote the population-level Hamilton filter of trend extraction as $\mathbb{H}_{h,p}$, which is a diagonal operator. After implementing the Hamilton filter and saving the residuals, that is, applying the operator
$(I_{N+1} - \mathbb{H}_{h,p}   )$ 
on both sides of \eqref{var_h}, we obtain:
\begin{eqnarray}
    c_{t+h} & = & (I_{N+1} - \mathbb{H}_{h,p}   ) Y_{t} 
     =  \zeta^{(h)}_{t} + \tilde  \Lambda^{(h)} \tilde f_t + \tilde  \epsilon_t^{(h)}, \label{var_res}
\end{eqnarray}
where 
$$
\zeta^{(h)}_{t} \equiv (I_{N+1} - \mathbb{H}_{h,p}   ) \sum_{j=1}^p B_j^{(h)} Y_{t-j} - \mathbb{H}_{h,p}  ( \tilde  \Lambda^{(h)} \tilde  f_t + \tilde  \epsilon_t^{(h)}).
$$
If the Hamilton filter is sufficient to capture the trend of each time series, 
then $c_{t+h}$ on the left-hand side of \eqref{var_res} is stationary. 
The second term of $\zeta^{(h)}_{t}$ is negligible  
uniformly over $t$ for it is a regression of stationary time series on nonstationary ones.
In particular, under the diagonal specification \eqref{eq:b_diag}, 
the first term of $\zeta_t$ vanishes as
$\mathbb{H}_{h,p} B_j Y_{t-j} =  B_j \mathbb{H}_{h,p} Y_{t-j} = B_j Y_{t-j} $.
This is the ideal data generating process that suits our procedure.

A generic VAR allows $B_j$ to be non-diagonal. In this case, the first item of $\zeta^{(h)}_{t}$ must remain stationary; otherwise, the two sides of \eqref{var_res}
--- a stationary left-hand side and a nonstationary right-hand side --- 
cannot be balanced. 
Given that $\zeta^{(h)}_{t}$ is stochastically bounded,  
we can further decompose 
it into an additive form of latent factors and idiosyncratic shocks, and merge them with $\tilde 
 f_t$ and $\tilde  \epsilon^{(h)}_t$, respectively, to form $f_t$ and $\epsilon_t$ in \eqref{asm:c-factor}.
In other words, we can derive the factor representation \eqref{asm:c-factor} 
from \eqref{var_h} even if $B_j$ are generic VAR coefficient matrices,
provided that the Hamilton filter is capable of removing the trend. 
In machine learning terminology, though
the data generating process as in \eqref{var_h}  is in a 
multivariate \emph{fully corrected} structure, it is sufficient to apply the Hamilton filter in a univariate \emph{recurrent} form, as shown in Figure \ref{fig:diag}, which is much more parsimonious than the fully connected version.

\section{Monte Carlo Simulations} \label{sec:simulation}

We conduct simulation studies to assess the accuracy of counterfactual imputation using the proposed synthetic business cycle approach relative to the conventional synthetic control method.

We consider three models. In each model, the total number of units $N+1$ is set to 12, mimicking the empirical study of Hong Kong’s return.\footnote{We also conducted simulations with five units, as in \cite{masini2022counterfactual}; in that scenario, the synthetic business cycle estimator performs even better.} Pre-treatment lengths $T_0$ are chosen from $\{50,100,200\}$. The forecasting horizon is $h=2$, and the number of autoregressive lags is $p=2$.

\paragraph{Model 1:} Independent random walks with drifts \citep[the first simulation model in][]{masini2022counterfactual}: for $i=1,\dots,N+1,\;\; t=1,\dots,T,$
\begin{align*}
    Y_{i,t}(0) \;=\; Y_{i,t-1}(0) \;+\; \mu_i \;+\; \varepsilon_{i,t},
\end{align*}
where $\varepsilon_{i,t}$'s are independent across both units and time, and are drawn from the standard normal distribution. The drift terms $\mu_i$ are either fixed at 0 or 0.5, or drawn from $\mathcal{N}(0,1/4)$. 

\paragraph{Model 2:} Idiosyncratic unit root trends with common stationary factors: for $i=1,\dots,N+1,\;\; t=1,\dots,T,$
\begin{align*}
    Y_{i,t}(0) & \;=\; Y_{i,t-1}(0) \;+\; \lambda^{ar}_{i,1}f^{ar}_{t,1} \;+\; \lambda^{ar}_{i,2}f^{ar}_{t,2} \;+\; \varepsilon_{i,t}, \\
    f^{ar}_{t,j} & \;=\; \phi f^{ar}_{t-1,j} \;+\; u^{ar}_{t,j}, j=1,2, 
\end{align*}
where $\varepsilon_{i,t}$'s and $u^{ar}_{t,j}$'s are independent standard normals, where the superscript ``\textit{ar}'' denotes ``autoregressive.'' The two factor loadings, $\lambda^{ar}_{i,1}$ and $\lambda^{ar}_{i,2}$, are drawn independently from the standard normal distribution. The autoregressive coefficient $\phi \in \{0.2,0.5,0.8\}.$ Even though this model includes a common factor structure, it generates no cointegrating relationships. \footnote{We suppress the drift term $\mu_i$ in Models 2 and 3 to highlight the leading order of the time series determined by the stochastic trend. The effect of the drift is small on the MSE ratio as it can be accurately estimated thanks to its very strong signal.
}

\paragraph{Model 3:} A partial cointegration structure: for $i=1,\cdots,\lfloor(N+1)/2\rfloor$, $t=1,\dots,T,$
\begin{align*}
    Y_{i,t}(0) & \;=\; \lambda^{rw}_{i,1}f^{rw}_{t,1} \;+\; 
    \lambda^{rw}_{i,2}f^{rw}_{t,2} \;+\; \lambda^{ar}_{i,1}f^{ar}_{t,1} \;+\; 
    \lambda^{ar}_{i,2}f^{ar}_{t,2} \;+\; \varepsilon_{i,t}, \\
    f^{rw}_{t,j} & \;=\; f^{rw}_{t-1,j} \;+\; u^{rw}_{t,j}, j=1,2, \\
    f^{ar}_{tj} & \;=\; \phi f^{ar}_{t-1,j} \;+\; u^{ar}_{t,j}, j=1,2,
\end{align*}
where $\varepsilon_{i,t}$'s, $u^{rw}_{t,j}$'s, and $u^{ar}_{t,j}$'s are independent standard normal, where the superscript ``\textit{rw}'' denotes ``random walk.'' The remaining units, $i=\lfloor (N+1)/2 \rfloor,\ldots ,N+1$, are simulated according to Model 2.  
Under this design, half of the units (including the treated unit) are cointegrated.
Their stationary-factor loadings, $\lambda^{ar}_{i,1}$ and $\lambda^{ar}_{i,2}$, are drawn independently from a standard normal distribution, as in Model 2.  The nonstationary-factor loadings, $\lambda^{rw}_{i,1}$ and $\lambda^{rw}_{i,2}$, are drawn independently from $\mathcal{N}(0,\,T_0^{-1/3})$.  With this choice, the nonstationary signal asymptotically dominates  while in finite sample the stationary factor plays a nontrivial role.   

In implementing the procedures, we consider three coefficient specifications in the vertical regression step. The first case, unrestricted coefficients, includes an intercept and places no constraints on the coefficients; the conventional synthetic control estimator under this specification corresponds to the procedure studied by \citet{masini2022counterfactual}. The second case, signed weights, omits the intercept and requires that the weights sum to one.
The third case, non-negative weights, adds the additional restriction that all weights be non-negative.


Table \ref{tab:sim} reports the simulation results from 10,000 Monte Carlo replications.  For each design, we compute the ratio of mean-squared errors (MSE) of the synthetic-business-cycle (SBC) estimator to those of the conventional synthetic-control (SC) estimator, separately for the pre- and post-treatment periods:
\begin{align*}
    \text{pre-treatment MSE ratio:} & \left( \sum_{t=h+p}^{T_0} (\hat{Y}_{1,t}^{\text{SBC}} -Y_{1,t} )^2 \right) \Big/ \left( \sum_{t=h+p}^{T_0} (\hat{Y}_{1,t}^{\text{SC}} -Y_{1,t} )^2 \right), \\
    \text{post-treatment MSE ratio:} & \left( \sum_{t=T_0+1}^{T_0+h} (\hat{Y}_{1,t}^{\text{SBC}} -Y_{1,t} )^2 \right) \Big/ \left( \sum_{t=T_0+1}^{T_0+h} (\hat{Y}_{1,t}^{\text{SC}} -Y_{1,t} )^2 \right),
\end{align*}where $\hat Y_{1,t}^{\text{SBC}}$ and $\hat Y_{1,t}^{\text{SC}}$ denote the SBC and SC counterfactual predictions, respectively. A ratio below 1 indicates that the synthetic business cycle estimator delivers a lower MSE than the conventional synthetic control estimator.

\begin{sidewaystable}[htbp!]
    \centering
    \begin{tabular}{@{}rcccccccccccccccccc@{}}
    \toprule
    \multicolumn{19}{c}{Panel (a): Independent random walks with drifts (Model 1)} \\ \midrule
    & \multicolumn{6}{c}{Unrestricted coefficients} & \multicolumn{6}{c}{Signed weights} & \multicolumn{6}{c}{Non-negative weights} \\
    \cmidrule(lr){2-7} \cmidrule(lr){8-13} \cmidrule(lr){14-19}
    & \multicolumn{2}{c}{$\mu_i=0$} & \multicolumn{2}{c}{$=0.5$} & \multicolumn{2}{c}{$\sim N(0,1/4)$} & \multicolumn{2}{c}{$\mu_i=0$} & \multicolumn{2}{c}{$=0.5$} & \multicolumn{2}{c}{$\sim N(0,1/4)$} & \multicolumn{2}{c}{$\mu_i=0$} & \multicolumn{2}{c}{$=0.5$} & \multicolumn{2}{c}{$\sim N(0,1/4)$} \\
    \cmidrule(lr){2-3} \cmidrule(lr){4-5}  \cmidrule(lr){6-7} \cmidrule(lr){8-9} \cmidrule(lr){10-11} \cmidrule(lr){12-13} \cmidrule(lr){14-15} \cmidrule(lr){16-17} \cmidrule(lr){18-19}
    & pre & post & pre & post & pre & post & pre & post & pre & post & pre & post & pre & post & pre & post & pre & post \\ \midrule
    $T_0=50$ & 1.10 & 0.64 & 1.10 & 0.65 & 1.10 & 0.64 & 0.89 & 0.54 & 0.92 & 0.56 & 0.93 & 0.56 & 0.24 & 0.15 & 0.25 & 0.16 & 0.06 & 0.03 \\
    $100$ & 0.72 & 0.38 & 0.72 & 0.36 & 0.73 & 0.37 & 0.59 & 0.31 & 0.61 & 0.32 & 0.60 & 0.31 & 0.15 & 0.08 & 0.15 & 0.08 & 0.02 & 0.01 \\
    $200$ & 0.41 & 0.21 & 0.41 & 0.21 & 0.41 & 0.21 & 0.34 & 0.17 & 0.34 & 0.17 & 0.34 & 0.17 & 0.08 & 0.04 & 0.08 & 0.04 & 0.01 & 0.00 \\ \midrule 
    \multicolumn{19}{c}{Panel (b): Idiosyncratic trends with common stationary factors (Model 2)} \\ \midrule
    & \multicolumn{6}{c}{Unrestricted coefficients} & \multicolumn{6}{c}{Signed weights} & \multicolumn{6}{c}{Non-negative weights} \\
    \cmidrule(lr){2-7} \cmidrule(lr){8-13} \cmidrule(lr){14-19}
    & \multicolumn{2}{c}{$\phi=0.2$} & \multicolumn{2}{c}{$0.5$} & \multicolumn{2}{c}{$0.8$} & \multicolumn{2}{c}{$\phi=0.2$} & \multicolumn{2}{c}{$0.5$} & \multicolumn{2}{c}{$0.8$} & \multicolumn{2}{c}{$\phi=0.2$} & \multicolumn{2}{c}{$0.5$} & \multicolumn{2}{c}{$0.8$} \\
    \cmidrule(lr){2-3} \cmidrule(lr){4-5}  \cmidrule(lr){6-7} \cmidrule(lr){8-9} \cmidrule(lr){10-11} \cmidrule(lr){12-13} \cmidrule(lr){14-15} \cmidrule(lr){16-17} \cmidrule(lr){18-19}
    & pre & post & pre & post & pre & post & pre & post & pre & post & pre & post & pre & post & pre & post & pre & post \\ \midrule
    $T_0=50$ & 1.23 & 0.76 & 1.52 & 0.83 & 2.31 & 1.15 & 1.02 & 0.62 & 1.24 & 0.69 & 1.86 & 0.94 & 0.26 & 0.18 & 0.31 & 0.20 & 0.43 & 0.26 \\
    $100$  & 0.77 & 0.41 & 0.91 & 0.45 & 1.38 & 0.65 & 0.63 & 0.34 & 0.73 & 0.37 & 1.10 & 0.52 & 0.16 & 0.08 & 0.19 & 0.09 & 0.27 & 0.12 \\
    $200$  & 0.42 & 0.21 & 0.49 & 0.23 & 0.73 & 0.33 & 0.35 & 0.18 & 0.40 & 0.20 & 0.59 & 0.28 & 0.09 & 0.04 & 0.10 & 0.04 & 0.15 & 0.06 \\ \midrule 
    \multicolumn{19}{c}{Panel (c): Partial cointegration (Model 3)} \\ \midrule
    & \multicolumn{6}{c}{Unrestricted coefficients} & \multicolumn{6}{c}{Signed weights} & \multicolumn{6}{c}{Non-negative weights} \\
    \cmidrule(lr){2-7} \cmidrule(lr){8-13} \cmidrule(lr){14-19}
    & \multicolumn{2}{c}{$\phi=0.2$} & \multicolumn{2}{c}{$0.5$} & \multicolumn{2}{c}{$0.8$} & \multicolumn{2}{c}{$\phi=0.2$} & \multicolumn{2}{c}{$0.5$} & \multicolumn{2}{c}{$0.8$} & \multicolumn{2}{c}{$\phi=0.2$} & \multicolumn{2}{c}{$0.5$} & \multicolumn{2}{c}{$0.8$} \\
    \cmidrule(lr){2-3} \cmidrule(lr){4-5}  \cmidrule(lr){6-7} \cmidrule(lr){8-9} \cmidrule(lr){10-11} \cmidrule(lr){12-13} \cmidrule(lr){14-15} \cmidrule(lr){16-17} \cmidrule(lr){18-19}
    & pre & post & pre & post & pre & post & pre & post & pre & post & pre & post & pre & post & pre & post & pre & post \\ \midrule
    $T_0=50$ & 1.05 & 0.98 & 1.04 & 0.94 & 1.09 & 0.91 & 0.97 & 0.94 & 0.97 & 0.89 & 1.01 & 0.85 & 0.75 & 0.65 & 0.69 & 0.56 & 0.48 & 0.29 \\
    $100$ & 1.05 & 0.99 & 1.02 & 0.94 & 1.03 & 0.89 & 0.95 & 0.86 & 0.92 & 0.82 & 0.93 & 0.77 & 0.73 & 0.56 & 0.66 & 0.47 & 0.43 & 0.26 \\
    $200$ & 1.03 & 0.97 & 0.99 & 0.95 & 0.97 & 0.85 & 0.92 & 0.85 & 0.88 & 0.80 & 0.84 & 0.74 & 0.69 & 0.49 & 0.62 & 0.42 & 0.41 & 0.24 
    \\ \bottomrule
    \end{tabular}
    \caption{Simulation results: MSE ratio }
    \label{tab:sim}
    \caption*{\footnotesize The table reports the ratio of the mean squared error for the synthetic business cycle estimator to that for the conventional synthetic control estimator, shown separately for the pre- and post-treatment periods. }
\end{sidewaystable}

In the unrestricted coefficients specification, Models 1 and 2 involve spurious regressions. The synthetic business cycle estimator consistently outperforms the conventional synthetic control estimator, achieving substantial MSE reductions once the pre-treatment sample reaches 100 or 200 observations. With only 50 pre-treatment observations, however, finite-sample issues slightly reverse this ranking. In Model 3—where half of the donor units are cointegrated with the treated unit, a setting for which the method of \citet{masini2022counterfactual} is specifically tailored—the synthetic business cycle estimator nonetheless matches or modestly outperforms in terms of MSE.

When we impose sum-to-unity and non-negativity constraints on the weights, the contrast becomes even sharper. Across all three models, the synthetic business cycle estimator yields much lower MSE, often reducing it by a factor of ten in the spurious-regression cases of Models 1 and 2. Even in the cointegrated setting of Model 3, it achieves reductions in MSE of 30\%–70\%. This indicates that the detrending procedure may make the pseudo-true weights closer to the simplex.

Finally, the impact of the autoregressive parameter $\phi$ differs between Models 2 and 3. In Model 2, a larger $\phi$ increases the persistence of the stationary factor, bringing the design closer to cointegration and slightly improving the conventional estimator, yet the synthetic business cycle estimator retains its advantage. In Model 3, when $\phi$ is small, both methods handle nonstationarity similarly well; when $\phi$ is large, the conventional one-step synthetic control ignores information in the stationary component, so the synthetic business cycle estimator remains more efficient.

These simulation exercises show the robustness in finite samples under both the spurious regression and cointegration. The filtering step effectively estimates the trend component. Compared to the other methods, we have a second ensemble learning step which further digests the information from the stationary component.

\section{Empirical Applications} \label{sec:empirical}

This section revisits two well-known empirical settings commonly analyzed using synthetic control and vertical regression methods: the economic impact of German reunification in 1990 and the return of Hong Kong to China in 1997. In both applications, the outcome of interest is per capita GDP, a nonstationary time series that our proposed methodology is specifically designed to address.

\subsection{German Reunification}
We revisit the seminal comparative case study by \cite{abadie2015comparative} on the economic impact of the 1990 German reunification on West Germany. Following the fall of the Berlin Wall in November 1989, the German Democratic Republic and the Federal Republic of Germany officially reunified on October 3, 1990. At the time of reunification, per capita GDP in West Germany was approximately 3 times higher than that of East Germany. 



\cite{abadie2015comparative} used 16 OECD countries as control units to construct a synthetic West Germany.\footnote{The 16 countries are Australia, Austria, Belgium, Denmark, France, Greece, Italy, Japan, the Netherlands, New Zealand, Norway, Portugal, Spain, Switzerland, the United Kingdom, and the United States.} The analysis is based on annual country-level data spanning 1960 to 2003. The pre-reunification period (1971–1990) was used to estimate the synthetic weights, while the post-reunification period (1991–2003) was used to evaluate the economic impact of German reunification on West Germany’s GDP. Their study finds that reunification had a substantial negative effect on West Germany’s GDP per capita. 

Figure \ref{fig:Germany_gdp_trend_cyc}(a) plots the raw GDP per capita of West Germany alongside that of the 16 OECD donor countries from 1960 to 1990. The Germany’s trajectory lies well within the donor envelope and shows no systematic divergence from the peer group throughout the pre-unification period. Most OECD countries, including West Germany, exhibit exponential GDP growth, though some experience a slowdown in growth rates as they approach 1990.

Figure \ref{fig:Germany_gdp_trend_cyc} (b) and (c) present the corresponding trend and cyclical components obtained with the Hamilton filter using pre-unification data. Both plots show a similar visual impression from the raw series: West Germany’s long‑run trend and short‑run fluctuations track those of several donor countries closely, lending further support to the appropriateness of a synthetic control design for assessing the effects of reunification. However, as we stress in this paper, the apparent comovement of countries’ GDP trends in the raw data may be spurious and should be interpreted with caution.

\begin{figure}[htbp!]
  \centering
  \begin{subfigure}{0.5\textwidth}
    \centering
\includegraphics[width=\textwidth]{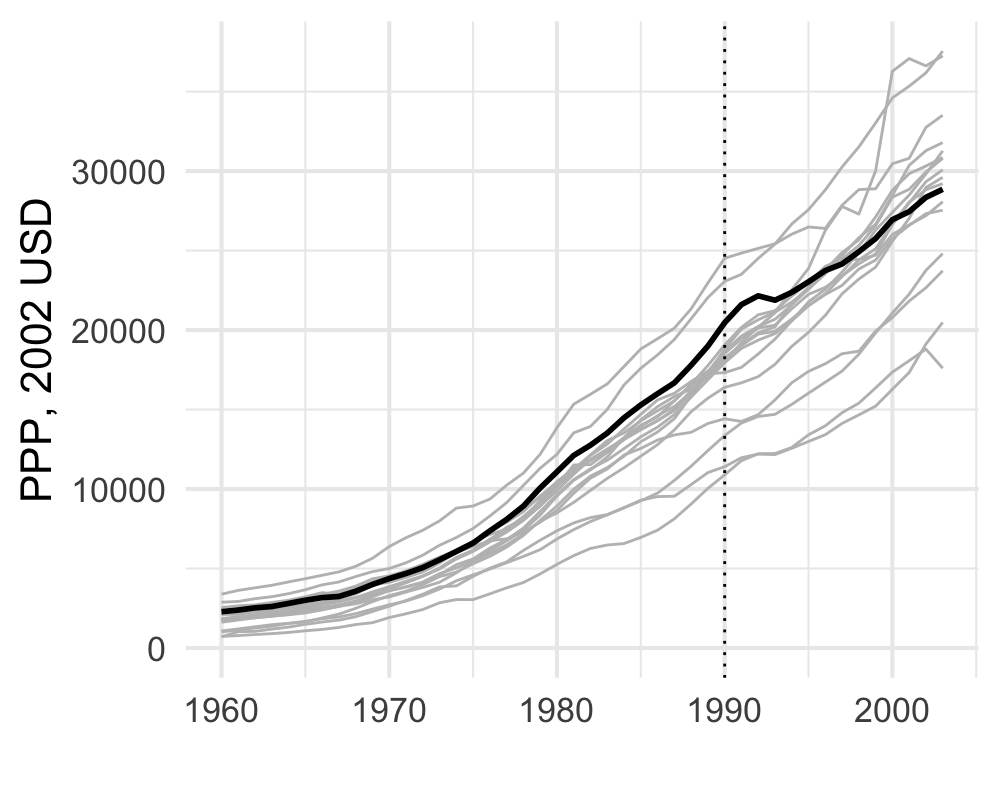}
    \caption{Raw per capita GDP}
  \end{subfigure}


  \begin{subfigure}{0.49\textwidth}
    \centering
 \includegraphics[width=\textwidth]{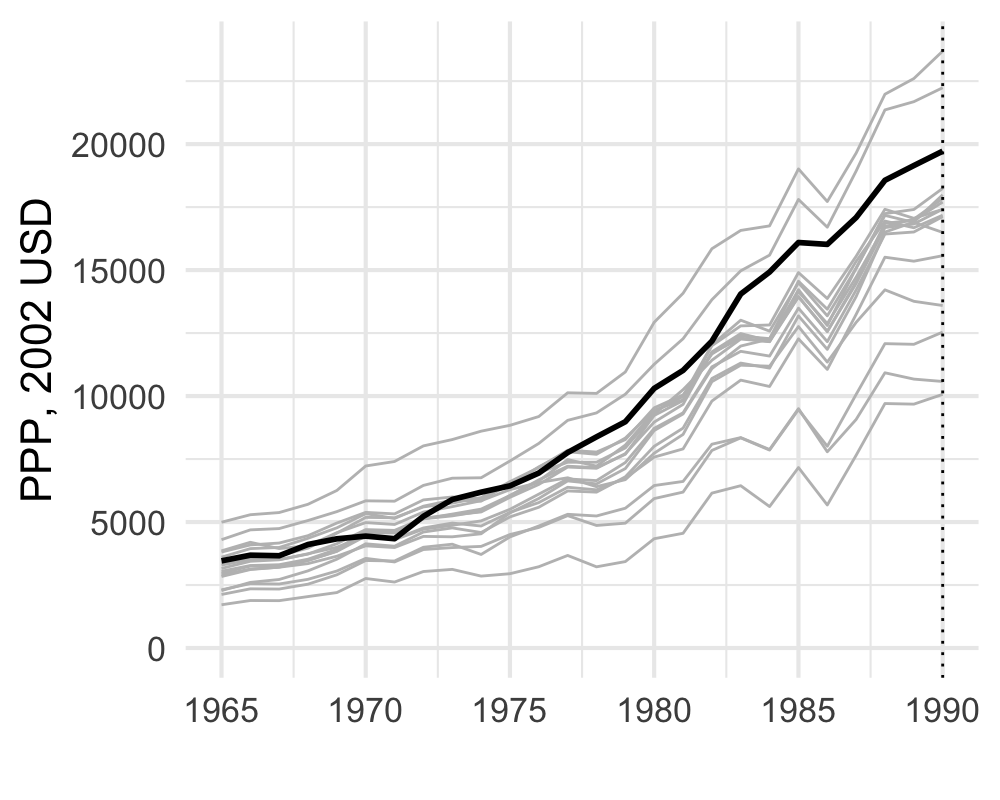}
    \caption{Trend components}
  \end{subfigure}
  \hfill
  \begin{subfigure}{0.49\textwidth}
    \centering
\includegraphics[width=\textwidth]{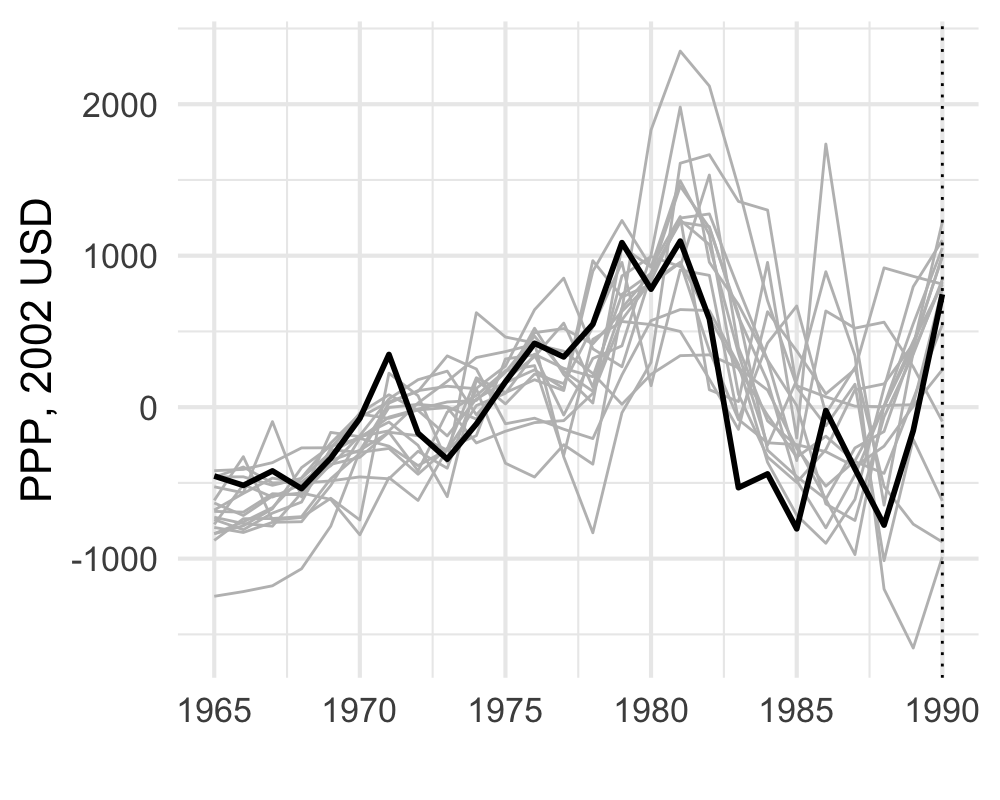}
    \caption{Cyclical components}
  \end{subfigure}
\caption{GDP path decomposition for Germany and donor countries}
  \label{fig:Germany_gdp_trend_cyc}
  \caption*{\footnotesize The thick black curve represents (West) Germany’s GDP per capita, while the grey curves correspond to the donor countries. Panel (a) shows the raw data; panels (b) and (c) display the trend and cyclical components, respectively, extracted with the Hamilton filter applied to the pre-unification data, using horizon $h=4$ years and $p=2$ lags.
}
\end{figure}

We apply the proposed synthetic business cycle (SBC) estimator and, for comparison, a conventional synthetic control (SC) estimator constructed from the raw GDP series. Note that this baseline estimator differs from the richer specification in \citet{abadie2015comparative}, which derives the weights based not only on pre-treatment outcomes but also on additional growth predictors such as inflation, trade openness, schooling, and industry share. Consequently, although the two approaches yield qualitatively similar patterns, our ``conventional synthetic control'' results are not identical to those of \citet{abadie2015comparative}; the divergence reflects our focus on the more recent matrix completion perspective of synthetic control \citep{athey2021matrix} rather than on a full covariate match.

Figure \ref{fig:GermanySBC} shows the per capita GDP of West Germany before and after the 1990 reunification, along with the two synthetic trajectories. Each method provides a four-year forecast of West Germany’s economy under the counterfactual scenario in which reunification did not occur after 1990. There are several findings. First, both estimators fit the pre-reunification data closely; although the conventional model matches the raw series slightly better, this tighter fit may reflect spurious comovements with the donor pool. Second, the synthetic business cycle estimator indicates a more substantial negative economic impact of reunification than the conventional estimator. Third, both approaches reveal a brief initial boost to per capita GDP after reunification, but the synthetic business cycle method confines this boom to roughly one year, whereas the conventional model extends it to three years. The shorter boom period aligns more closely with earlier studies \citep{MeinhardtEtAl1995}.


\begin{figure}[htbp!]
  \centering
  \begin{subfigure}[b]{0.49\textwidth}
    \includegraphics[width=\textwidth]{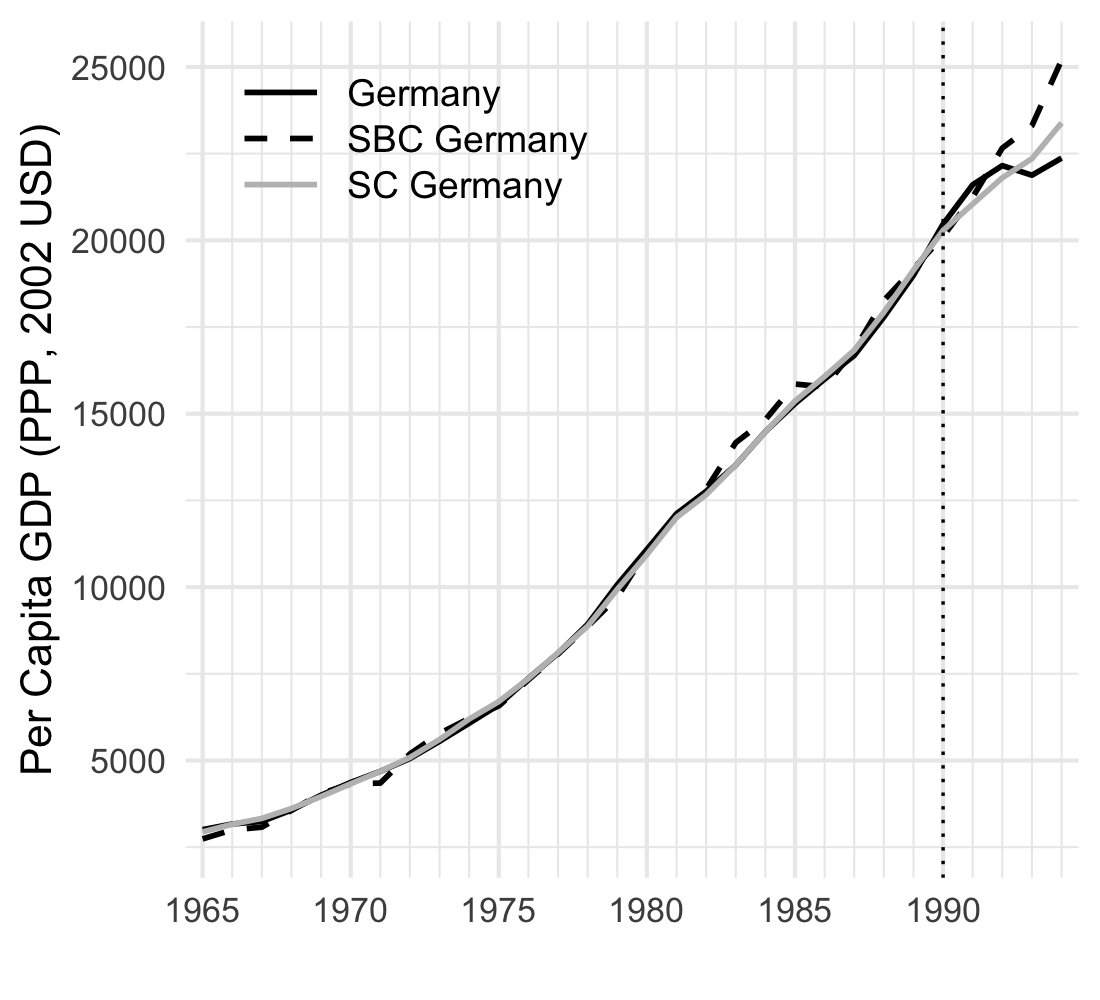}
    \caption{Synthetic and actual GDP}
  \end{subfigure}
  \hfill
  \begin{subfigure}[b]{0.49\textwidth}
    \includegraphics[width=\textwidth]{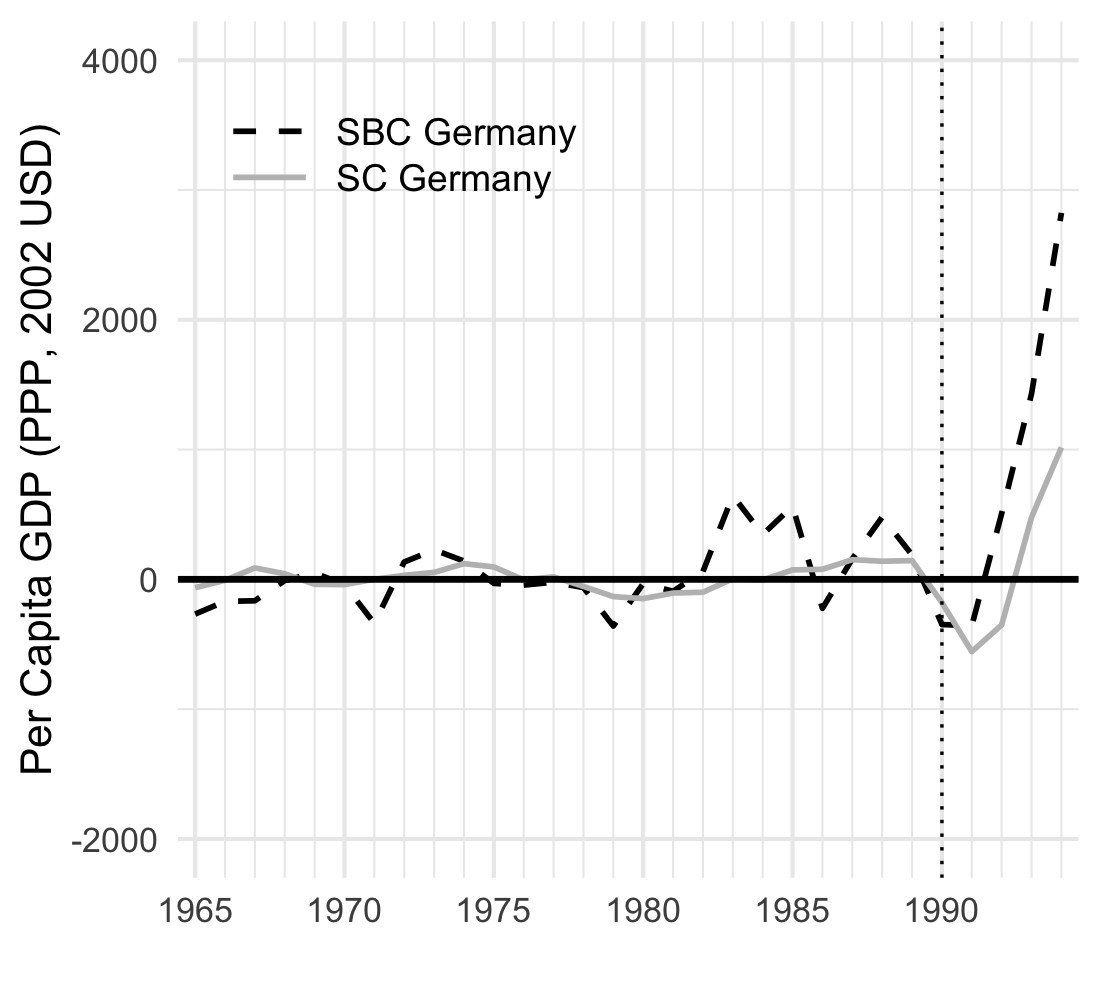}
    \caption{Deviation from actual GDP}
  \end{subfigure}
\caption{Treatment effect estimates of German reunification}
  \label{fig:GermanySBC}
  \caption*{\footnotesize Panel (a) shows Germany’s synthetic GDP estimated using both the synthetic business cycle (SBC) estimator and the conventional synthetic control (SC) estimator, compared with the actual GDP. Panel (b) shows the differences between the two synthetic GDP trajectories and the actual GDP.}
\end{figure}

\begin{figure}[h]
    \centering
    \includegraphics[width=0.9\linewidth]{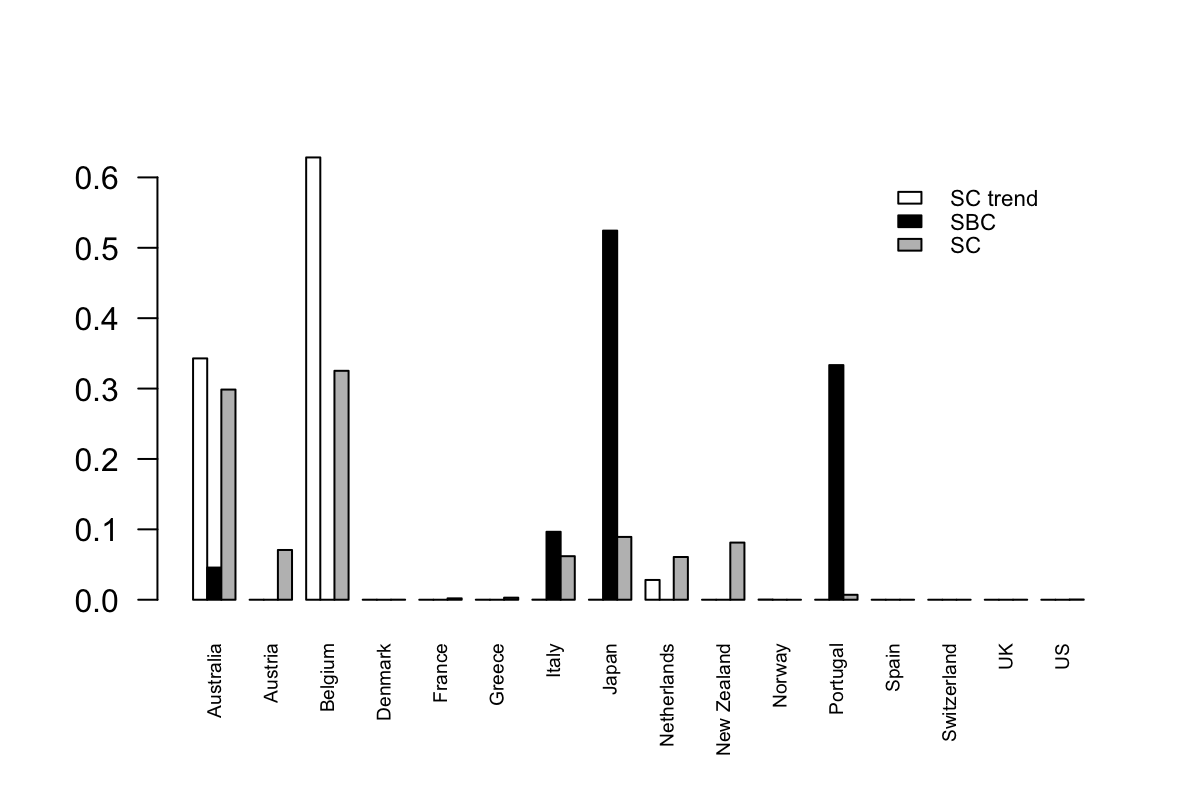}
    \caption{Comparison of synthetic weights}
    \label{fig:GermanyWeightsComparison}
    \caption*{\footnotesize Weights assigned to donor countries using three different approaches. SC (grey) represents the synthetic control method applied to the raw GDP series. SC trend (white) and SBC (black) represent synthetic control applied to the decomposed trend and cyclical components, respectively.}
\end{figure}

To investigate the source of the discrepancy between the standard synthetic control and the synthetic business cycle estimators, we compare the donor weights generated by the two procedures. Figure \ref{fig:GermanyWeightsComparison} reports three sets of weights: (i) the conventional synthetic control weights based on the raw series; (ii) the weights assigned to the cyclical component by the synthetic business cycle method; and (iii) a ``synthetic trend'' weight vector that best reproduces West Germany’s trend component using the donors’ trends.

Examining these weights reveals a sharp contrast. The synthetic trend specification places nearly all of its weight on Australia and Belgium, indicating that these two economies most closely match West Germany’s long-run growth path. By contrast, the synthetic business cycle weights fall mainly on Italy, Japan, and Portugal, whose short-run fluctuations align more closely with West Germany’s business cycle. When applied to the unfiltered data, the conventional synthetic control effectively averages across these distinct structures: it assigns weight both to the trend matching donors (Australia and Belgium) and to the cycle matching donors (Italy, Japan, and Portugal), with the trend countries receiving larger weights because the trend component dominates the cyclical component in magnitude.

These patterns imply that West Germany’s trend and cyclical components correspond to different donor subsets, reflecting two distinct factor structures. The conventional approach, which presumes a single common structure and applies a single set of weights, risks combining spurious trend comovements with genuine cyclical similarities. Our synthetic business cycle estimator avoids this pitfall by forecasting West Germany’s trend from its own history and pairing it with a synthetic cycle drawn from the donor pool. This separation of trend and cycle yields a more robust counterfactual prediction.

In Figure \ref{fig:Germany_placebo}, we perform the placebo exercise as in \cite{abadie2015comparative}. We re‑estimate both the conventional synthetic control and our synthetic business cycle method, artificially shifting the reunification date to 1975, fifteen years before the actual event. With this placebo reunification timing, the synthetic business cycle forecasts remain closely aligned with the observed data throughout 1975–1990, whereas the forecasts from the conventional synthetic control estimator exhibit a modest downward drift.\footnote{By contrast, the richer specification in \citet{abadie2015comparative}, which incorporates additional non‑outcome predictors, performs well in their placebo test.} Viewed from the matrix‑completion perspective, these results suggest that the synthetic business cycle estimator can offer a more robust alternative to the conventional synthetic‑control approach when only outcome paths are used for causal inference.


\begin{figure}[htbp!]
  \centering
  \begin{subfigure}[b]{0.49\textwidth}
    \includegraphics[width=\textwidth]{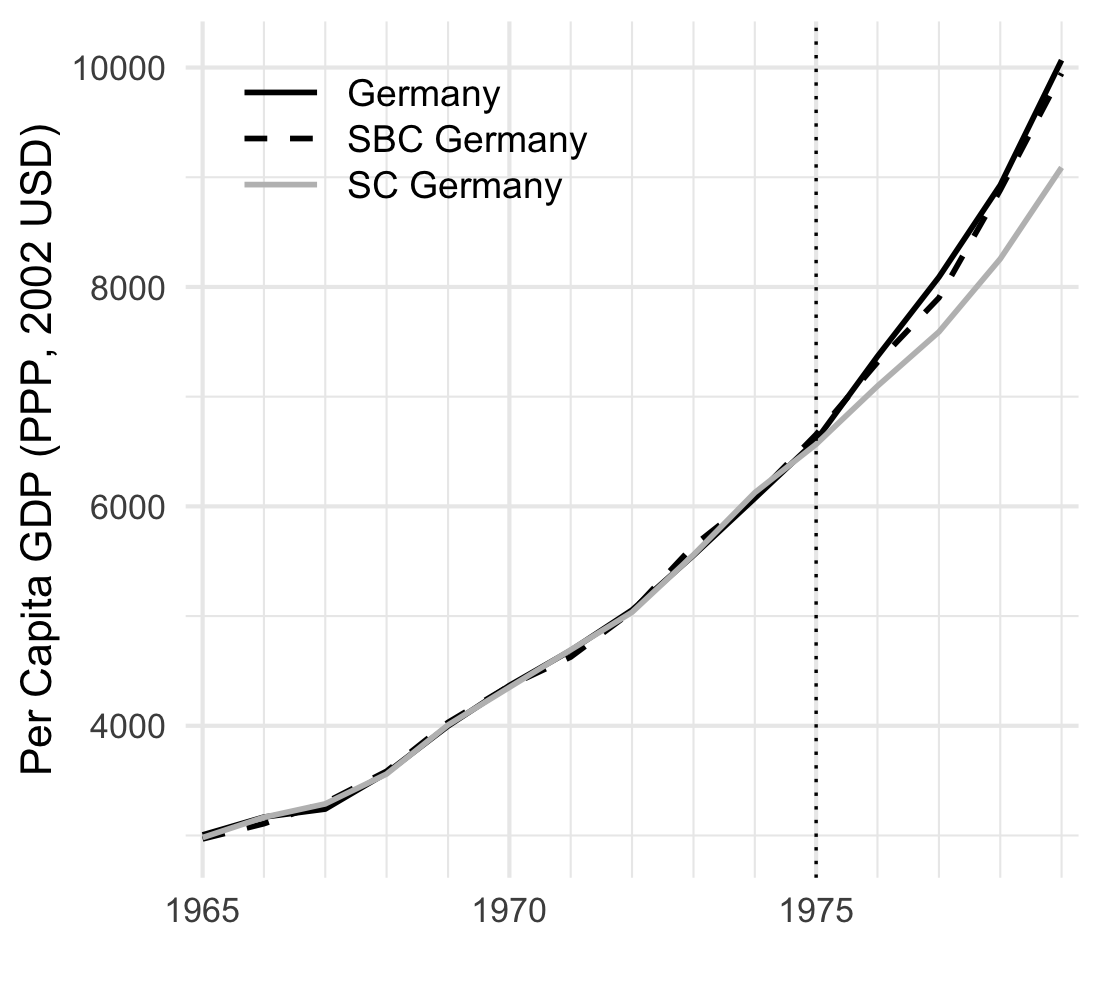}
    \caption{Synthetic and actual GDP}
  \end{subfigure}
  \hfill
  \begin{subfigure}[b]{0.49\textwidth}
    \includegraphics[width=\textwidth]{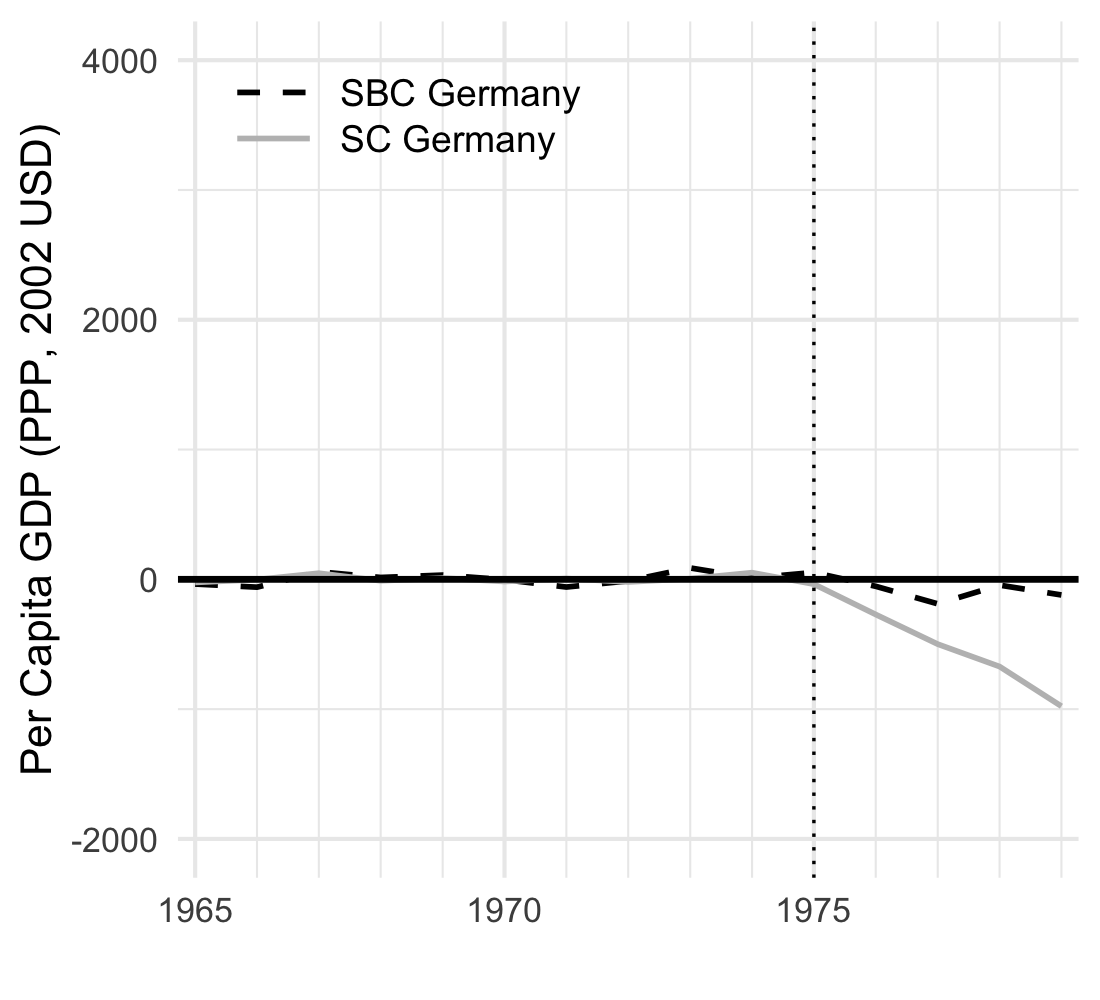}
    \caption{Deviation from actual GDP}
  \end{subfigure}
  \caption{Placebo reunification in 1975}
  \label{fig:Germany_placebo}
  \caption*{\footnotesize This figure applies the synthetic business cycle (SBC) and the conventional synthetic control (SC) method to the placebo reunification dated 1975. Panel (a) shows Germany’s synthetic GDP estimated using both the synthetic business cycle estimator and the conventional synthetic control estimator, compared with the actual GDP. Panel (b) shows the differences between the two synthetic GDP trajectories and the actual GDP.}
\end{figure}

\subsection{Return of Hong Kong to China} \label{sec:HongKong}

China resumed sovereignty over Hong Kong on July 1, 1997, ending British colonial rule. \cite{hsiao2012panel} used the vertical regression method to study the effect on Hong Kong's economic growth rate, which is a stationary series. Here, we examine the raw GDP series as the outcome, similar to the previous empirical exercise, highlighting the synthetic business cycle's power of handling nonstationary outcomes.

We work with annual GDP levels for 1961–2003 sourced from Federal Reserve Economic Data (FRED).\footnote{Our dataset differs from that in \citet{hsiao2012panel}, who use quarterly GDP‑growth data that start later.} Because the pre-treatment window is short, we confine the donor pool to 11 developed, market‑oriented economies—Australia, Austria, Canada, Denmark, France, Germany, Italy, Korea, the Netherlands, New Zealand, and the US—that exhibit similar long‑run GDP trajectories to Hong Kong.\footnote{These economies are drawn from the original list of 23 in \citet{hsiao2012panel}. We also omit Finland, Norway, and Switzerland, whose high per‑capita GDP and welfare‑state structures make them less comparable with Hong Kong.} The UK and mainland China are excluded, as both were directly involved in Hong Kong’s 1997 return and are therefore unsuitable for counterfactual analysis.

Our selection differs from that of \citet{hsiao2012panel}, who emphasize geographic proximity and economic ties and select as donors mainland China, Indonesia, Japan, Korea, Malaysia, the Philippines, Singapore, Taiwan, Thailand, and the US. We adopt a different strategy for two reasons. First, economies that are closely linked to Hong Kong may themselves have been influenced by the 1997 return, weakening their validity as controls. Second, several Southeast Asian economies were hit hard by the 1997–1998 Asian financial crisis, introducing confounding shocks that could obscure the effect of the return. As a robustness check, we present an analysis using the donor economies used in \cite{hsiao2012panel} in Appendix \ref{sec:additional-empirical}.

\begin{figure}[htbp!]
  \centering
  \begin{subfigure}{0.5\textwidth}
    \centering
\includegraphics[width=\textwidth]{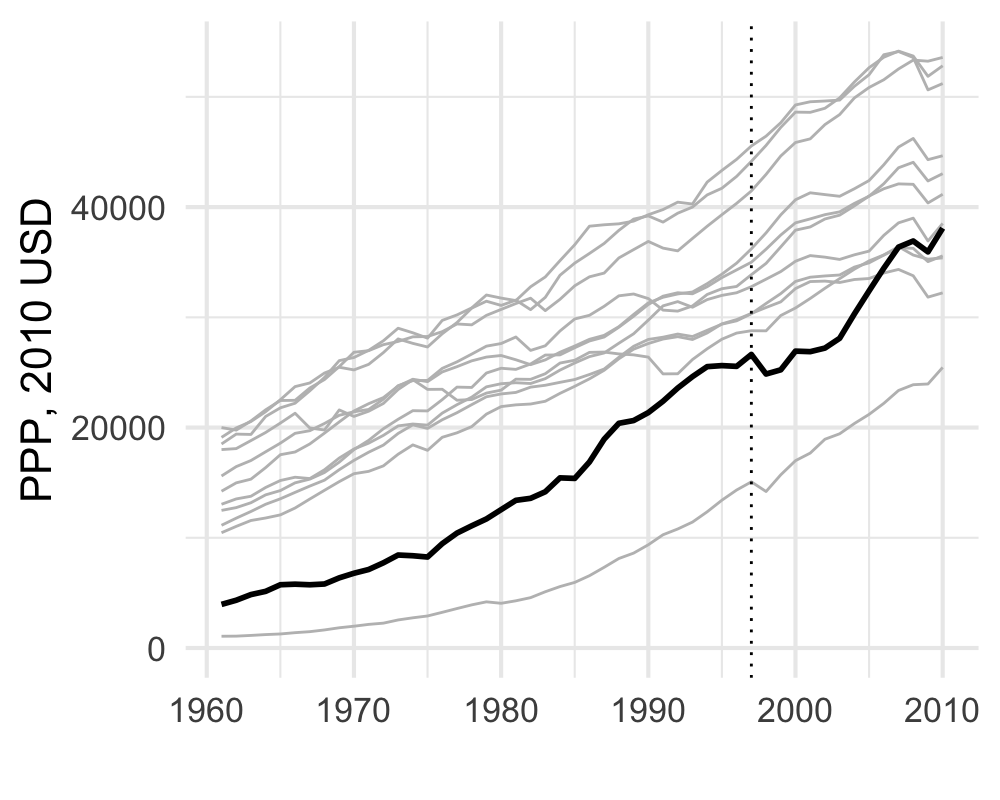}
    \caption{Raw per capita GDP}
  \end{subfigure}


  \begin{subfigure}{0.49\textwidth}
    \centering
 \includegraphics[width=\textwidth]{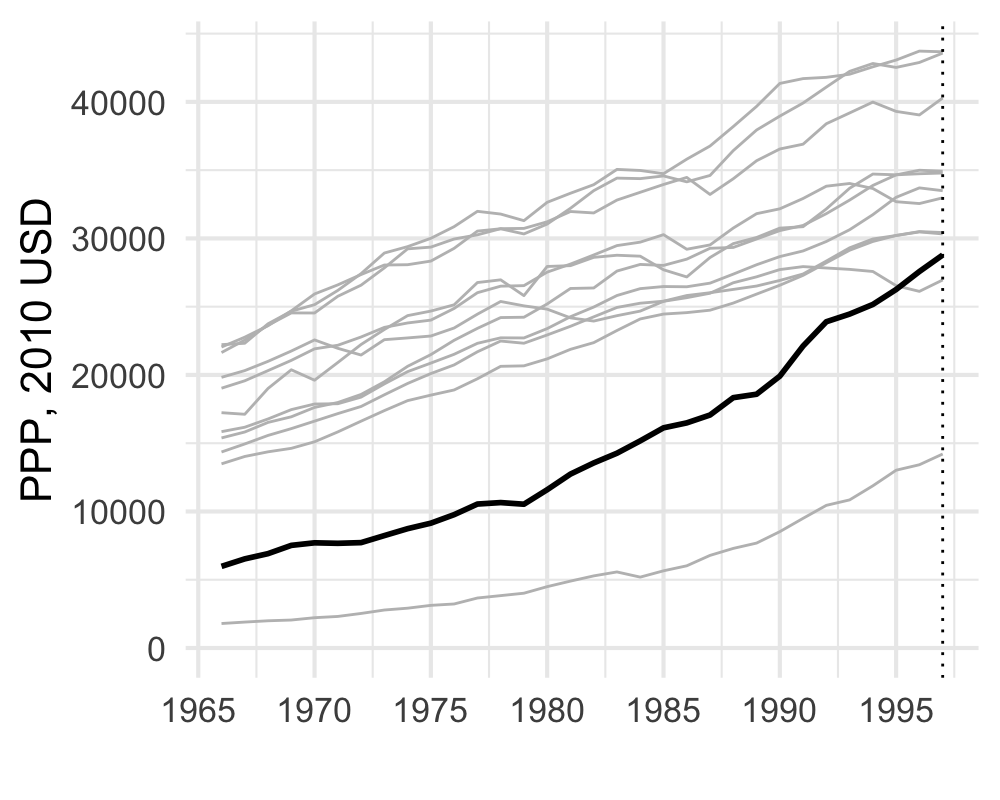}
    \caption{Trend components}
  \end{subfigure}
  \hfill
  \begin{subfigure}{0.49\textwidth}
    \centering
\includegraphics[width=\textwidth]{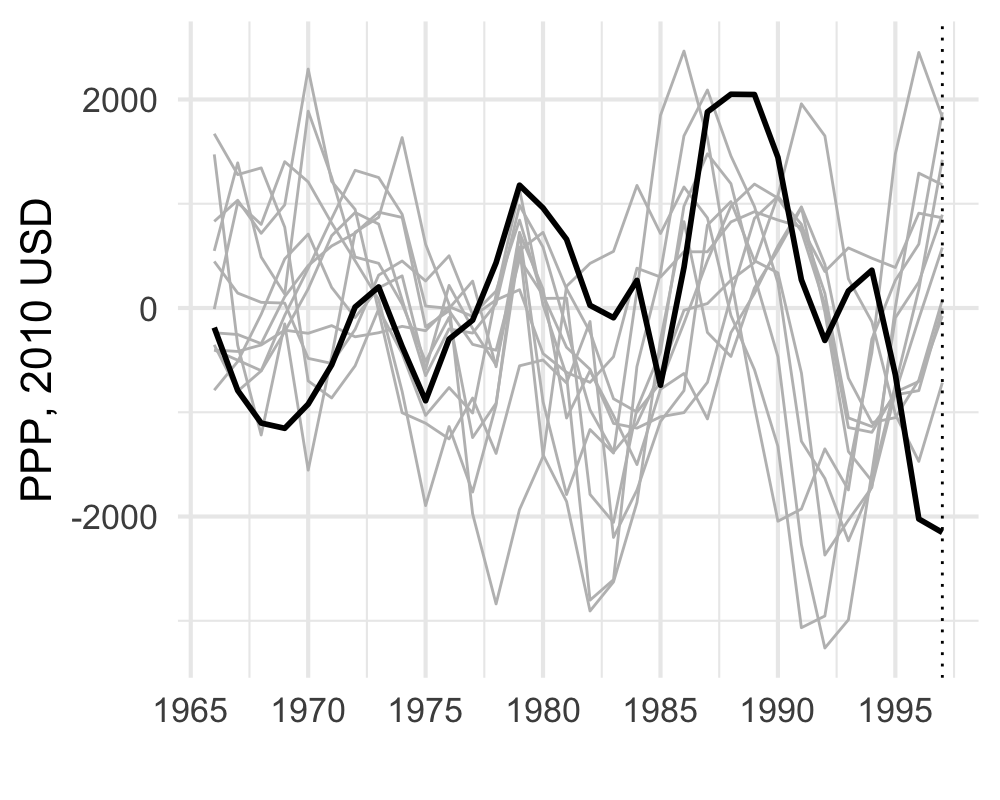}
    \caption{Cyclical components}
  \end{subfigure}
  \caption{GDP path decomposition for Hong Kong and donor countries}
  \label{fig:HK_gdp_trend_cyc}
  \caption*{\footnotesize The thick black curve represents Hong Kong's GDP per capita, while the grey curves correspond to the donor countries. Panel (a) shows the raw data; panels (b) and (c) display the trend and cyclical components, respectively, extracted with the Hamilton filter applied to the pre-unification data, using horizon $h=4$ years and $p=2$ lags.
}
\end{figure}

Figure \ref{fig:HK_gdp_trend_cyc} plots Hong Kong’s per capita GDP alongside that of the eleven donor economies, together with the corresponding trend and cyclical components. Because \citet{hsiao2012panel} employ vertical regression, placing no restrictions on the weights, we implement the analysis under two specifications: with and without the non-negativity constraint on the weights.\footnote{In both configurations, the weights still sum to one. This constraint is required to handle a common nonstationary factor with identical loadings across all units. See the data-generating process in Appendix B of \cite{abadie2010synthetic}.} Relaxing the non‑negativity constraint is sensible, as Figure \ref{fig:HK_gdp_trend_cyc}(c) shows that some donors’ cycles are negatively correlated with Hong Kong’s.

\begin{figure}[htbp!]
  \centering
  \begin{subfigure}{0.49\textwidth}
    \includegraphics[width=\textwidth]{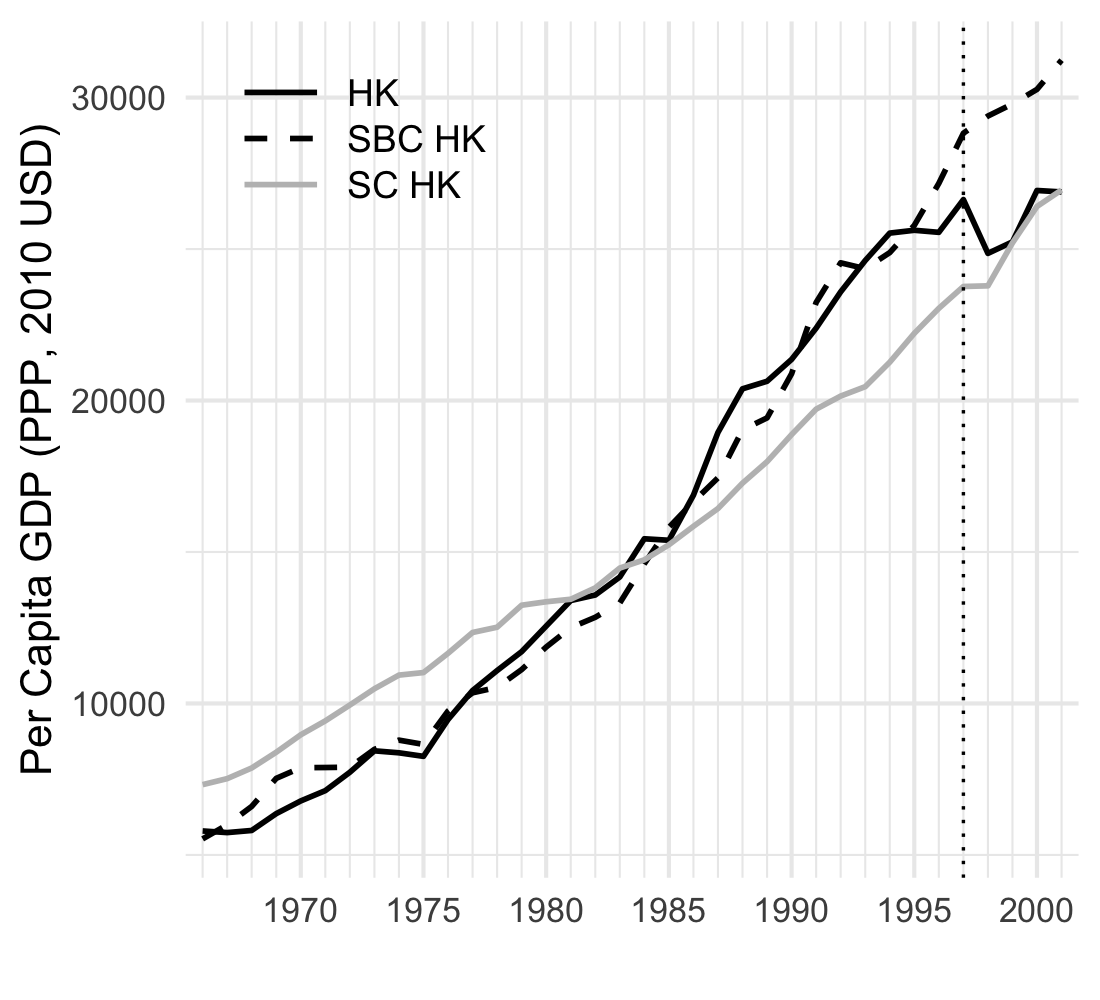}
    \caption{Non-negative weights}
  \end{subfigure}
  \hfill
  \begin{subfigure}{0.49\textwidth}
    \includegraphics[width=\textwidth]{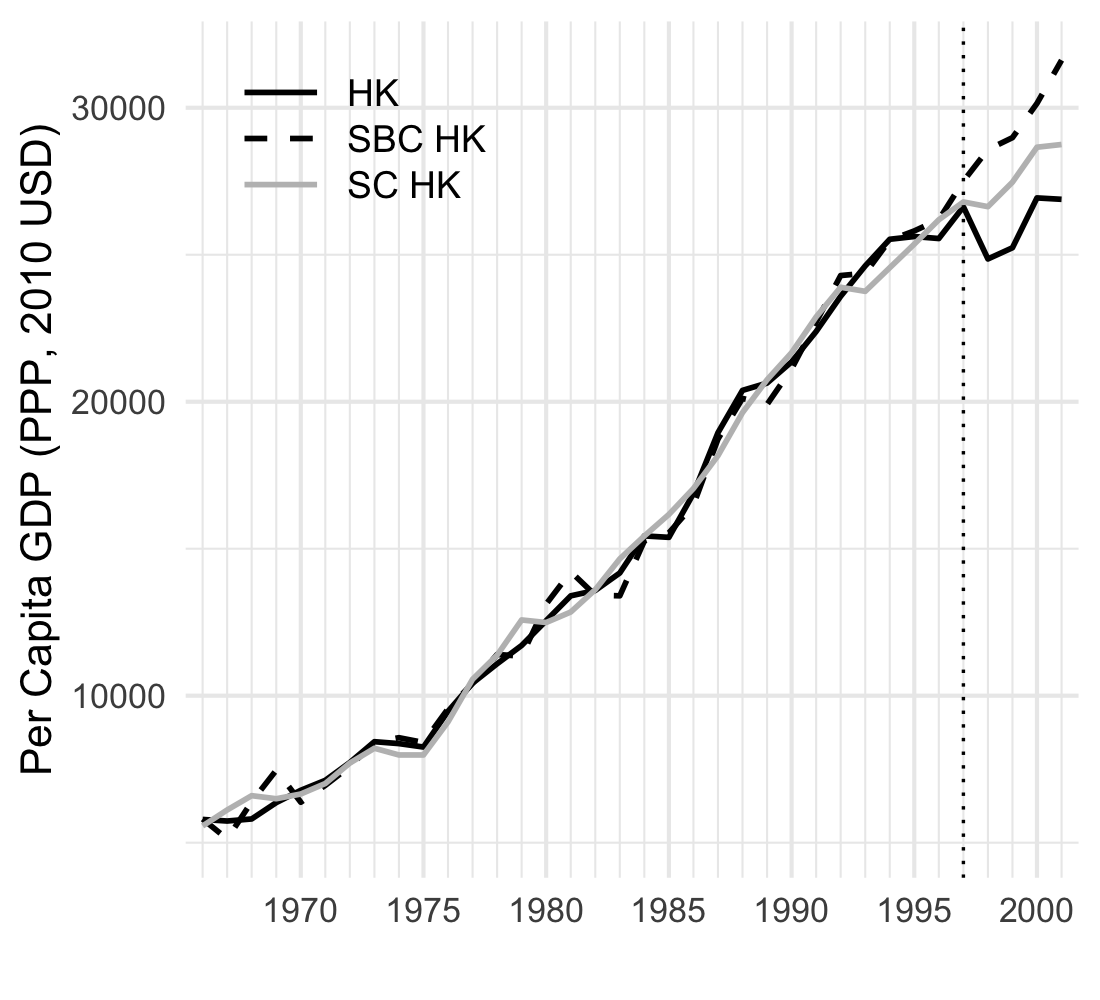}
    \caption{Signed weighted}
  \end{subfigure}
  \caption{Treatment effect estimates of return of Hong Kong}
  \label{fig:hk-treatment-effect}
  \caption*{\footnotesize Hong Kong's synthetic GDP estimated using both the synthetic business cycle estimator and the conventional synthetic control estimator, compared with the actual GDP. Panels (a) and (b) impose and relax the non-negativity constraint on the weights, respectively.}
\end{figure}

Figure \ref{fig:hk-treatment-effect} shows the counterfactual predictions from the synthetic business cycle estimator and the conventional synthetic control estimator. Two main findings emerge. First, the synthetic business cycle estimator provides a good pre-treatment fit using both non-negative and signed weights, whereas the conventional synthetic control estimator fits well only with signed weights, highlighting the greater flexibility of the synthetic business cycle approach. Second, the synthetic business cycle estimator yields a higher estimate of the economic cost of the handover.

Figure \ref{fig:hk-weights} compares the weights for the raw data, trend, and cycle components. Consistent with the previous empirical analysis, we find that the synthetic control weights for the raw data approximate a combination of the trend and cycle weights. Notably, when weights are not constrained to be non-negative, the trend and cycle weights can have opposite signs, resulting in overall synthetic weights for the raw data that are small in magnitude due to offsetting effects.
 
\begin{figure}[htbp!]
  \centering
  \begin{subfigure}{0.49\textwidth}
    \includegraphics[width=\textwidth]{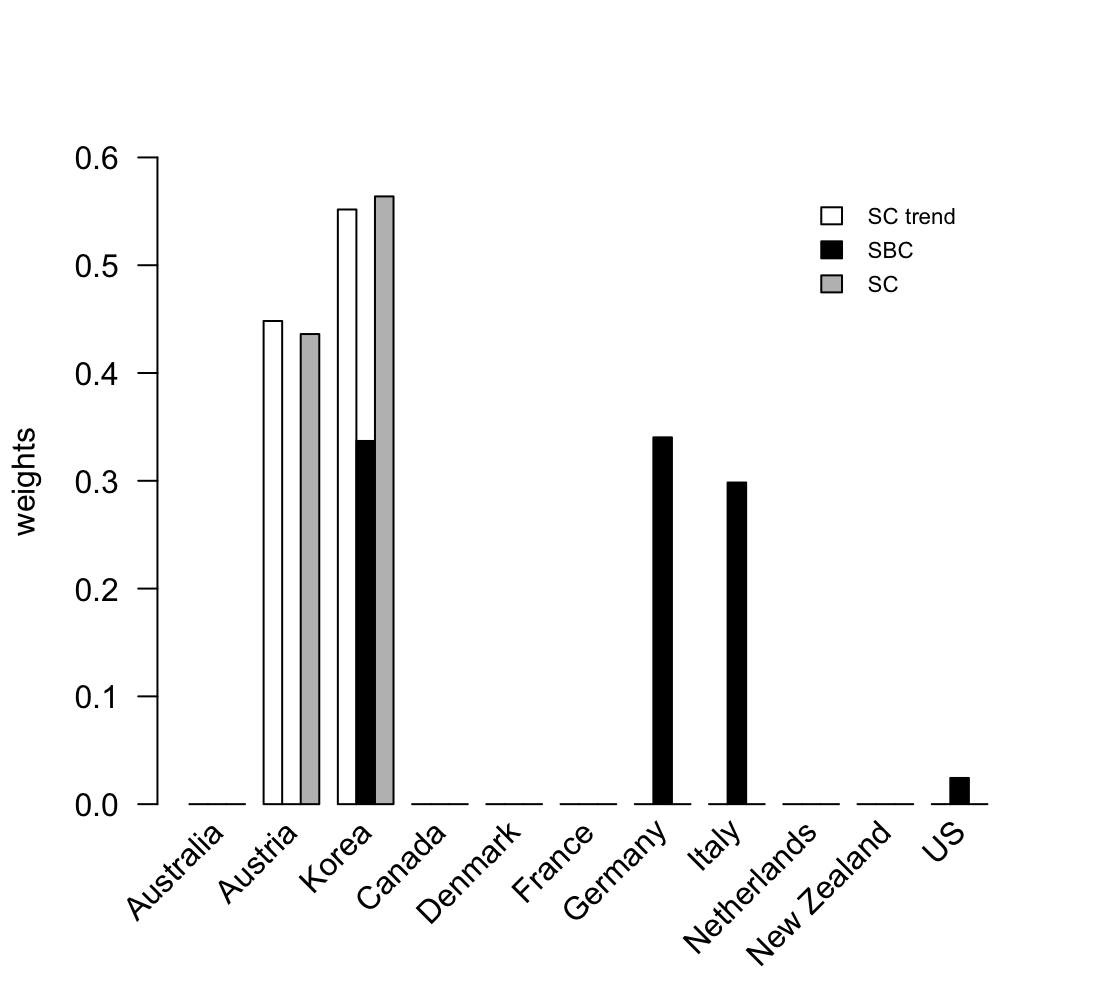}
    \caption{Non-negative weights}
  \end{subfigure}
  \hfill
  \begin{subfigure}{0.49\textwidth}
    \includegraphics[width=\textwidth]{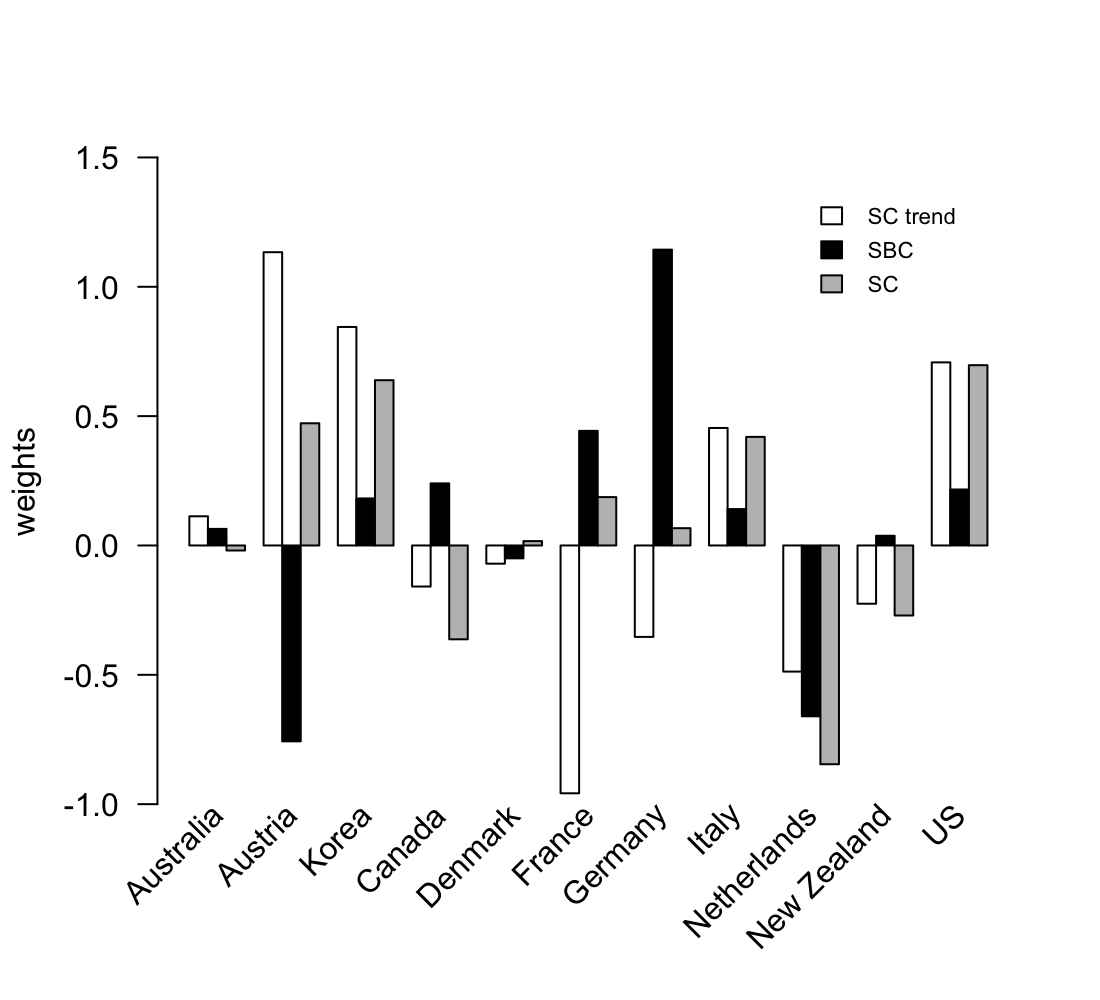}
    \caption{Signed weights}
  \end{subfigure}
  \caption{Comparison of synthetic weights}
  \label{fig:hk-weights}
  \caption*{\footnotesize Weights assigned to donor economies using three different approaches. SC (grey) represents the synthetic control method applied to the raw GDP series. SC trend (white) and SBC (black) represent synthetic control applied to the decomposed trend and cyclical components, respectively. Panels (a) and (b) impose and relax the non-negativity constraint on the weights, respectively.}
\end{figure}

Figure \ref{fig:hk-placebo} shows the placebo test with the treatment date set to 1987, ten years before the actual event. The synthetic business cycle forecasts remain closely aligned with the observed data throughout 1987–1997, whereas the forecasts from the conventional synthetic control estimator exhibit a downward deviation, which is more pronounced when weights are restricted.

\begin{figure}[htbp!]
  \centering
  \begin{subfigure}{0.49\textwidth}
    \includegraphics[width=\textwidth]{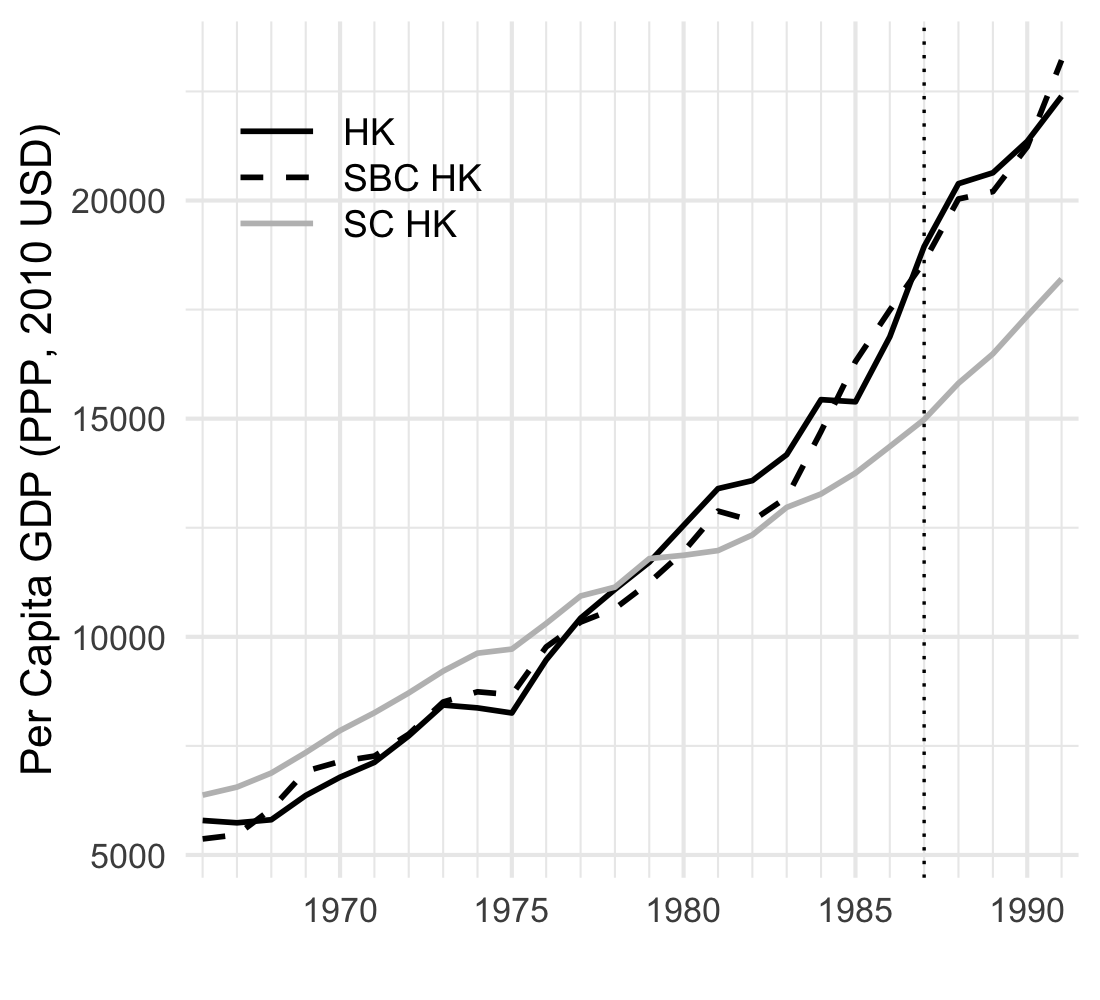}
    \caption{Non-negative weights}
  \end{subfigure}
  \hfill
  \begin{subfigure}{0.49\textwidth}
    \includegraphics[width=\textwidth]{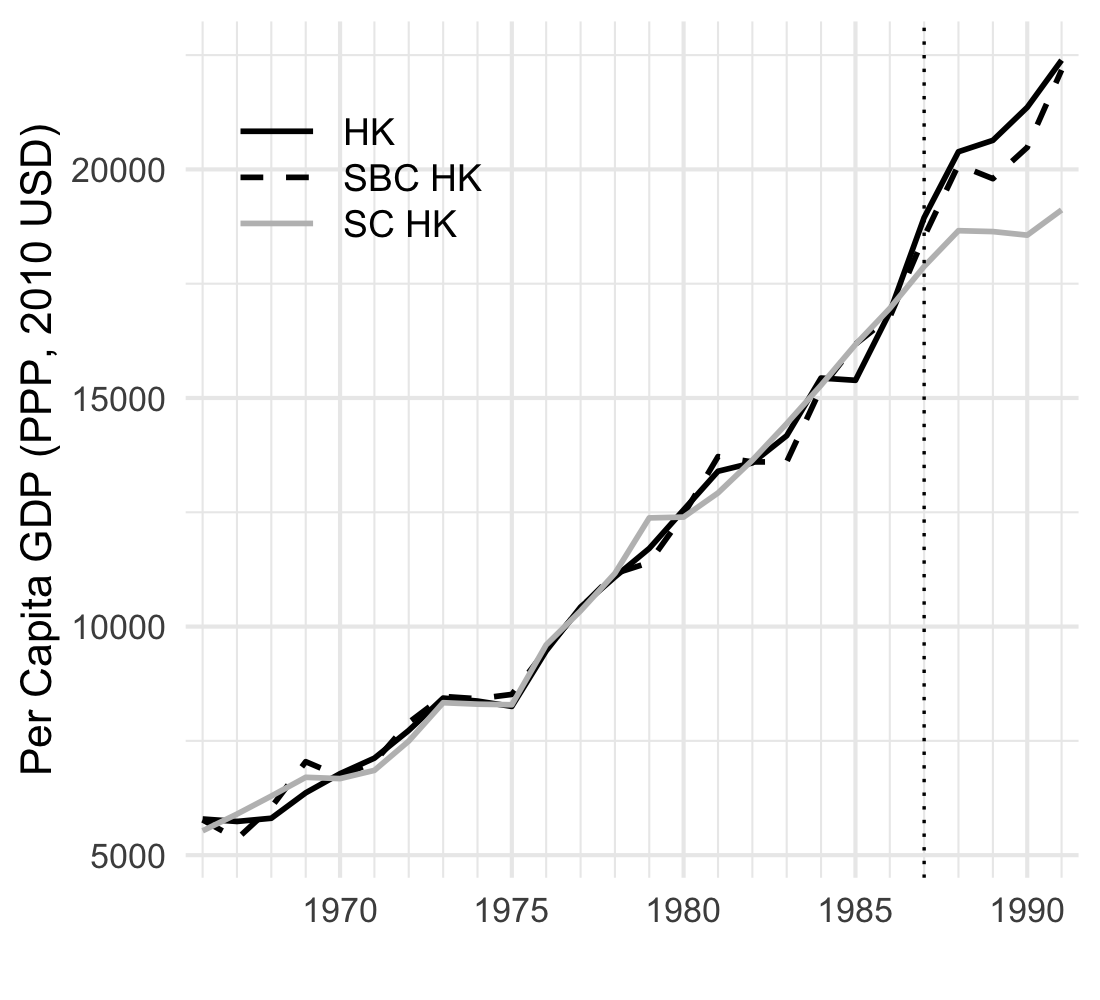}
    \caption{Signed weights}
  \end{subfigure}
  \caption{Placebo return of Hong Kong in 1987}
  \label{fig:hk-placebo}
  \caption*{\footnotesize This figure applies the synthetic business cycle and the conventional synthetic control method to the placebo return of Hong Kong dated 1987. Panels (a) and (b) impose and relax the non-negativity constraint on the weights, respectively.}
\end{figure}

\section{Conclusion} \label{sec:conclusion}

In this paper, we identify a spurious synthetic control problem: the conventional synthetic control estimator can spuriously attribute common movements to treatment effects when outcomes follow nonstationary trends. To remedy this, we introduce the synthetic business cycle estimator, which forecasts the treated unit’s trend from its own pre-treatment trend and its cycle from weighted donor cycles, and prove its asymptotic unbiasedness even with idiosyncratic trends. Empirical applications to German reunification and the return of Hong Kong confirm its superior accuracy in counterfactual analyses.

The synthetic business cycle estimator extends the panel causal inference toolkit and offers a reliable alternative for policy evaluation with nonstationary data. We recommend that empirical researchers use this method when conducting counterfactual analyses of potentially nonstationary economic data.

The framework in this paper can be extended in a variety of directions. First, for simplicity we follow the literature to consider the ``fixed-$n$, large-$T$'' asymptotics. It would be important to generalize it to the case when the number of cross-section units also diverges to infinity, which will require a delicate analysis of the relative magnitude of $n$ and $T$, and may further involve regularization schemes for dimensional reduction. Second, if the data possess either multiple treated units or a long post-treatment period, then it is possible to invoke a stable limit distribution to quantify the statistical uncertainty of the average treatment effect. The asymptotic distribution will count on the relative rates of convergence of the trend estimation and the linear ensemble of the estimated cyclical components.
These are interesting topics for future research.

\bigskip
\bigskip

\bibliographystyle{chicago}
\bibliography{sc}

\begin{thebibliography}{}

\bibitem[\protect\citeauthoryear{Abadie}{Abadie}{2021}]{abadie2021using}
Abadie, A. (2021).
\newblock Using synthetic controls: Feasibility, data requirements, and methodological aspects.
\newblock {\em Journal of Economic Literature\/}~{\em 59\/}(2), 391--425.

\bibitem[\protect\citeauthoryear{Abadie, Diamond, and Hainmueller}{Abadie et~al.}{2010}]{abadie2010synthetic}
Abadie, A., A.~Diamond, and J.~Hainmueller (2010).
\newblock Synthetic control methods for comparative case studies: Estimating the effect of {C}alifornia’s tobacco control program.
\newblock {\em Journal of the American statistical Association\/}~{\em 105\/}(490), 493--505.

\bibitem[\protect\citeauthoryear{Abadie, Diamond, and Hainmueller}{Abadie et~al.}{2015}]{abadie2015comparative}
Abadie, A., A.~Diamond, and J.~Hainmueller (2015).
\newblock Comparative politics and the synthetic control method.
\newblock {\em American Journal of Political Science\/}~{\em 59\/}(2), 495--510.

\bibitem[\protect\citeauthoryear{Abadie and Gardeazabal}{Abadie and Gardeazabal}{2003}]{abadie2003economic}
Abadie, A. and J.~Gardeazabal (2003).
\newblock The economic costs of conflict: A case study of the basque country.
\newblock {\em American Economic Review\/}~{\em 93\/}(1), 113--132.

\bibitem[\protect\citeauthoryear{Acemoglu, Johnson, Kermani, Kwak, and Mitton}{Acemoglu et~al.}{2016}]{acemoglu2016value}
Acemoglu, D., S.~Johnson, A.~Kermani, J.~Kwak, and T.~Mitton (2016).
\newblock The value of connections in turbulent times: Evidence from the united states.
\newblock {\em Journal of Financial Economics\/}~{\em 121\/}(2), 368--391.

\bibitem[\protect\citeauthoryear{Andersson}{Andersson}{2019}]{andersson2019carbon}
Andersson, J.~J. (2019).
\newblock Carbon taxes and {CO$_2$} emissions: Sweden as a case study.
\newblock {\em American Economic Journal: Economic Policy\/}~{\em 11\/}(4), 1--30.

\bibitem[\protect\citeauthoryear{Arkhangelsky, Athey, Hirshberg, Imbens, and Wager}{Arkhangelsky et~al.}{2021}]{arkhangelsky2021synthetic}
Arkhangelsky, D., S.~Athey, D.~A. Hirshberg, G.~W. Imbens, and S.~Wager (2021).
\newblock Synthetic difference-in-differences.
\newblock {\em American Economic Review\/}~{\em 111\/}(12), 4088--4118.

\bibitem[\protect\citeauthoryear{Asatryan, Castell{\'o}n, and Stratmann}{Asatryan et~al.}{2018}]{asatryan2018balanced}
Asatryan, Z., C.~Castell{\'o}n, and T.~Stratmann (2018).
\newblock Balanced budget rules and fiscal outcomes: Evidence from historical constitutions.
\newblock {\em Journal of Public Economics\/}~{\em 167}, 105--119.

\bibitem[\protect\citeauthoryear{Athey, Bayati, Doudchenko, Imbens, and Khosravi}{Athey et~al.}{2021}]{athey2021matrix}
Athey, S., M.~Bayati, N.~Doudchenko, G.~Imbens, and K.~Khosravi (2021).
\newblock Matrix completion methods for causal panel data models.
\newblock {\em Journal of the American Statistical Association\/}~{\em 116\/}(536), 1716--1730.

\bibitem[\protect\citeauthoryear{Avila-Montealegre and Mix}{Avila-Montealegre and Mix}{2024}]{avila2024common}
Avila-Montealegre, O. and C.~Mix (2024).
\newblock Common trade exposure and business cycle comovement.
\newblock {\em Journal of International Economics\/}~{\em 152}, 103998.

\bibitem[\protect\citeauthoryear{Bai, Li, and Ouyang}{Bai et~al.}{2014}]{bai2014property}
Bai, C., Q.~Li, and M.~Ouyang (2014).
\newblock Property taxes and home prices: A tale of two cities.
\newblock {\em Journal of Econometrics\/}~{\em 180\/}(1), 1--15.

\bibitem[\protect\citeauthoryear{Bai}{Bai}{2009}]{bai2009panel}
Bai, J. (2009).
\newblock Panel data models with interactive fixed effects.
\newblock {\em Econometrica\/}~{\em 77\/}(4), 1229--1279.

\bibitem[\protect\citeauthoryear{Bai and Ng}{Bai and Ng}{2021}]{bai2021matrix}
Bai, J. and S.~Ng (2021).
\newblock Matrix completion, counterfactuals, and factor analysis of missing data.
\newblock {\em Journal of the American Statistical Association\/}~{\em 116\/}(536), 1746--1763.

\bibitem[\protect\citeauthoryear{Becker, Pfeifer, and Schweikert}{Becker et~al.}{2021}]{becker2021price}
Becker, M., G.~Pfeifer, and K.~Schweikert (2021).
\newblock Price effects of the austrian fuel price fixing act: a synthetic control study.
\newblock {\em Energy Economics\/}~{\em 97}, 105207.

\bibitem[\protect\citeauthoryear{Ben-Michael, Feller, and Rothstein}{Ben-Michael et~al.}{2021}]{ben2021augmented}
Ben-Michael, E., A.~Feller, and J.~Rothstein (2021).
\newblock The augmented synthetic control method.
\newblock {\em Journal of the American Statistical Association\/}~{\em 116\/}(536), 1789--1803.

\bibitem[\protect\citeauthoryear{Beveridge and Nelson}{Beveridge and Nelson}{1981}]{beveridge1981new}
Beveridge, S. and C.~R. Nelson (1981).
\newblock A new approach to decomposition of economic time series into permanent and transitory components with particular attention to measurement of the ‘business cycle’.
\newblock {\em Journal of Monetary Economics\/}~{\em 7\/}(2), 151--174.

\bibitem[\protect\citeauthoryear{Billmeier and Nannicini}{Billmeier and Nannicini}{2013}]{billmeier2013assessing}
Billmeier, A. and T.~Nannicini (2013).
\newblock Assessing economic liberalization episodes: A synthetic control approach.
\newblock {\em Review of Economics and Statistics\/}~{\em 95\/}(3), 983--1001.

\bibitem[\protect\citeauthoryear{Botosaru and Ferman}{Botosaru and Ferman}{2019}]{botosaru2019role}
Botosaru, I. and B.~Ferman (2019).
\newblock On the role of covariates in the synthetic control method.
\newblock {\em The Econometrics Journal\/}~{\em 22\/}(2), 117--130.

\bibitem[\protect\citeauthoryear{Carvalho, Masini, and Medeiros}{Carvalho et~al.}{2018}]{carvalho2018arco}
Carvalho, C., R.~Masini, and M.~C. Medeiros (2018).
\newblock Ar{C}o: An artificial counterfactual approach for high-dimensional panel time-series data.
\newblock {\em Journal of Econometrics\/}~{\em 207\/}(2), 352--380.

\bibitem[\protect\citeauthoryear{Chamon, Garcia, and Souza}{Chamon et~al.}{2017}]{chamon2017fx}
Chamon, M., M.~Garcia, and L.~Souza (2017).
\newblock {FX} interventions in {B}razil: {A} synthetic control approach.
\newblock {\em Journal of International Economics\/}~{\em 108}, 157--168.

\bibitem[\protect\citeauthoryear{Di~Giovanni, Levchenko, and Mejean}{Di~Giovanni et~al.}{2018}]{di2018micro}
Di~Giovanni, J., A.~A. Levchenko, and I.~Mejean (2018).
\newblock The micro origins of international business-cycle comovement.
\newblock {\em American Economic Review\/}~{\em 108\/}(1), 82--108.

\bibitem[\protect\citeauthoryear{Doudchenko and Imbens}{Doudchenko and Imbens}{2016}]{doudchenko2016balancing}
Doudchenko, N. and G.~W. Imbens (2016).
\newblock Balancing, regression, difference-in-differences and synthetic control methods: A synthesis.
\newblock Technical report, National Bureau of Economic Research.

\bibitem[\protect\citeauthoryear{Eliason and Lutz}{Eliason and Lutz}{2018}]{eliason2018can}
Eliason, P. and B.~Lutz (2018).
\newblock Can fiscal rules constrain the size of government? {A}n analysis of the “crown jewel” of tax and expenditure limitations.
\newblock {\em Journal of Public Economics\/}~{\em 166}, 115--144.

\bibitem[\protect\citeauthoryear{Gobillon and Magnac}{Gobillon and Magnac}{2016}]{gobillon2016regional}
Gobillon, L. and T.~Magnac (2016).
\newblock Regional policy evaluation: Interactive fixed effects and synthetic controls.
\newblock {\em Review of Economics and Statistics\/}~{\em 98\/}(3), 535--551.

\bibitem[\protect\citeauthoryear{Granger and Newbold}{Granger and Newbold}{1974}]{granger1974spurious}
Granger, C.~W. and P.~Newbold (1974).
\newblock Spurious regressions in econometrics.
\newblock {\em Journal of Econometrics\/}~{\em 2\/}(2), 111--120.

\bibitem[\protect\citeauthoryear{Hamilton}{Hamilton}{2018}]{hamilton2018you}
Hamilton, J.~D. (2018).
\newblock Why you should never use the {H}odrick-{P}rescott filter.
\newblock {\em Review of Economics and Statistics\/}~{\em 100\/}(5), 831--843.

\bibitem[\protect\citeauthoryear{Hodrick and Prescott}{Hodrick and Prescott}{1997}]{hodrick1997postwar}
Hodrick, R.~J. and E.~C. Prescott (1997).
\newblock Postwar us business cycles: an empirical investigation.
\newblock {\em Journal of Money, Credit, and Banking\/}~{\em 29\/}(1), 1--16.

\bibitem[\protect\citeauthoryear{Hou, Li, Li, and Ouyang}{Hou et~al.}{2021}]{hou2021revisiting}
Hou, L., K.~Li, Q.~Li, and M.~Ouyang (2021).
\newblock Revisiting the location of fdi in china: A panel data approach with heterogeneous shocks.
\newblock {\em Journal of Econometrics\/}~{\em 221\/}(2), 483--509.

\bibitem[\protect\citeauthoryear{Hsiao, Ching, and Wan}{Hsiao et~al.}{2012}]{hsiao2012panel}
Hsiao, C., S.~H. Ching, and S.~K. Wan (2012).
\newblock A panel data approach for program evaluation: Measuring the benefits of political and economic integration of {H}ong {K}ong with mainland {C}hina.
\newblock {\em Journal of Applied Econometrics\/}~{\em 27\/}(5), 705--740.

\bibitem[\protect\citeauthoryear{Hsiao, Shi, and Zhou}{Hsiao et~al.}{2022}]{hsiao2022transformed}
Hsiao, C., Z.~Shi, and Q.~Zhou (2022).
\newblock Transformed estimation for panel interactive effects models.
\newblock {\em Journal of Business \& Economic Statistics\/}~{\em 40\/}(4), 1831--1848.

\bibitem[\protect\citeauthoryear{Ibragimov and Sharakhmetov}{Ibragimov and Sharakhmetov}{2002}]{ibragimov2002exact}
Ibragimov, R. and S.~Sharakhmetov (2002).
\newblock The exact constant in the rosenthal inequality for random variables with mean zero.
\newblock {\em Theory of Probability \& Its Applications\/}~{\em 46\/}(1), 127--132.

\bibitem[\protect\citeauthoryear{Li, Shen, and Zhou}{Li et~al.}{2024}]{li2024confidence}
Li, X., Y.~Shen, and Q.~Zhou (2024).
\newblock Confidence intervals of treatment effects in panel data models with interactive fixed effects.
\newblock {\em Journal of Econometrics\/}~{\em 240\/}(1), 105684.

\bibitem[\protect\citeauthoryear{Masini and Medeiros}{Masini and Medeiros}{2021}]{masini2021counterfactual}
Masini, R. and M.~C. Medeiros (2021).
\newblock Counterfactual analysis with artificial controls: Inference, high dimensions, and nonstationarity.
\newblock {\em Journal of the American Statistical Association\/}~{\em 116\/}(536), 1773--1788.

\bibitem[\protect\citeauthoryear{Masini and Medeiros}{Masini and Medeiros}{2022}]{masini2022counterfactual}
Masini, R. and M.~C. Medeiros (2022).
\newblock Counterfactual analysis and inference with nonstationary data.
\newblock {\em Journal of Business \& Economic Statistics\/}~{\em 40\/}(1), 227--239.

\bibitem[\protect\citeauthoryear{Mayoral}{Mayoral}{2013}]{mayoral2013heterogeneous}
Mayoral, L. (2013).
\newblock Heterogeneous dynamics, aggregation, and the persistence of economic shocks.
\newblock {\em International Economic Review\/}~{\em 54\/}(4), 1295--1307.

\bibitem[\protect\citeauthoryear{McCracken and Ng}{McCracken and Ng}{2016}]{mccracken2016fred}
McCracken, M.~W. and S.~Ng (2016).
\newblock {FRED-MD}: {A} monthly database for macroeconomic research.
\newblock {\em Journal of Business \& Economic Statistics\/}~{\em 34\/}(4), 574--589.

\bibitem[\protect\citeauthoryear{Mei, Phillips, and Shi}{Mei et~al.}{2024}]{mei2024boosted}
Mei, Z., P.~C. Phillips, and Z.~Shi (2024).
\newblock The boosted {H}odrick-{P}rescott filter is more general than you might think.
\newblock {\em Journal of Applied Econometrics\/}~{\em 39\/}(7), 1260--1281.

\bibitem[\protect\citeauthoryear{Meinhardt, Seidel, Stille, and Teichmann}{Meinhardt et~al.}{1995}]{MeinhardtEtAl1995}
Meinhardt, V., B.~Seidel, F.~Stille, and D.~Teichmann (1995).
\newblock Transferleistungen in die neuen bundesländer und deren wirtschaftliche konsequenzen.
\newblock Sonderheft 154, Deutsches Institut für Wirtschaftsforschung (DIW Berlin), Berlin.

\bibitem[\protect\citeauthoryear{Miao, Phillips, and Su}{Miao et~al.}{2023}]{miao2023high}
Miao, K., P.~C. Phillips, and L.~Su (2023).
\newblock High-dimensional vars with common factors.
\newblock {\em Journal of Econometrics\/}~{\em 233\/}(1), 155--183.

\bibitem[\protect\citeauthoryear{Moon and Weidner}{Moon and Weidner}{2017}]{moon2017dynamic}
Moon, H.~R. and M.~Weidner (2017).
\newblock Dynamic linear panel regression models with interactive fixed effects.
\newblock {\em Econometric Theory\/}~{\em 33\/}(1), 158--195.

\bibitem[\protect\citeauthoryear{Onatski and Wang}{Onatski and Wang}{2021}]{onatski2021spurious}
Onatski, A. and C.~Wang (2021).
\newblock Spurious factor analysis.
\newblock {\em Econometrica\/}~{\em 89\/}(2), 591--614.

\bibitem[\protect\citeauthoryear{Peri and Yasenov}{Peri and Yasenov}{2019}]{peri2019labor}
Peri, G. and V.~Yasenov (2019).
\newblock The labor market effects of a refugee wave: Synthetic control method meets the mariel boatlift.
\newblock {\em Journal of Human Resources\/}~{\em 54\/}(2), 267--309.

\bibitem[\protect\citeauthoryear{Phillips}{Phillips}{1986}]{phillips1986understanding}
Phillips, P.~C. (1986).
\newblock Understanding spurious regressions in econometrics.
\newblock {\em Journal of Econometrics\/}~{\em 33\/}(3), 311--340.

\bibitem[\protect\citeauthoryear{Phillips and Shi}{Phillips and Shi}{2021}]{phillips2021boosting}
Phillips, P.~C. and Z.~Shi (2021).
\newblock Boosting: {W}hy you can use the {HP} filter.
\newblock {\em International Economic Review\/}~{\em 62\/}(2), 521--570.

\bibitem[\protect\citeauthoryear{Shen, Ding, Sekhon, and Yu}{Shen et~al.}{2023}]{shen2023same}
Shen, D., P.~Ding, J.~Sekhon, and B.~Yu (2023).
\newblock Same root different leaves: Time series and cross-sectional methods in panel data.
\newblock {\em Econometrica\/}~{\em 91\/}(6), 2125--2154.

\bibitem[\protect\citeauthoryear{Shi and Huang}{Shi and Huang}{2023}]{shi2023forward}
Shi, Z. and J.~Huang (2023).
\newblock Forward-selected panel data approach for program evaluation.
\newblock {\em Journal of Econometrics\/}~{\em 234\/}(2), 512--535.

\bibitem[\protect\citeauthoryear{Stock and Watson}{Stock and Watson}{1988}]{stock1988variable}
Stock, J.~H. and M.~W. Watson (1988).
\newblock Variable trends in economic time series.
\newblock {\em Journal of Economic Perspectives\/}~{\em 2\/}(3), 147--174.

\bibitem[\protect\citeauthoryear{Stock and Watson}{Stock and Watson}{1999}]{stock1999business}
Stock, J.~H. and M.~W. Watson (1999).
\newblock Business cycle fluctuations in {US} macroeconomic time series.
\newblock {\em Handbook of Macroeconomics\/}~{\em 1}, 3--64.

\bibitem[\protect\citeauthoryear{Xiong and Pelger}{Xiong and Pelger}{2023}]{xiong2023large}
Xiong, R. and M.~Pelger (2023).
\newblock Large dimensional latent factor modeling with missing observations and applications to causal inference.
\newblock {\em Journal of Econometrics\/}~{\em 233\/}(1), 271--301.

\end{thebibliography}

\newpage

\appendix

\numberwithin{equation}{section}
\numberwithin{figure}{section}
\numberwithin{lemma}{section}

\section{Proofs for Theoretical Results}
Before we prove our main result, it is useful to first show the consistency of the synthetic weights using the estimated cyclical components. 

\begin{lemma}\label{lm:w-consistency}
Given Assumption \ref{asm:c-factor}-\ref{asm:error convergence}, the proposed synthetic business cycle estimator has the following convergent result: for $i = 2,\dots, N+1$,
\begin{align*}
    \hat{w}_i - w_i^0 \overset{p}{\rightarrow} c_i,
\end{align*}
where $c_i$'s are finite constants, and $w_i^0$'s are the infeasible synthetic weights if $\lambda_i$ and pre-treatment factors $f_t$ were directly observable, i.e., 
\begin{align*}
    \lambda_{1}f_t = \sum_{i = 2}^{N+1} w_i^0\lambda_if_t.
\end{align*}
\end{lemma}
Note that the infeasible synthetic weights $w_i^0$'s are finite given Assumption \ref{asm:c-factor}(i). We follow \cite{abadie2010synthetic} in assuming a perfect pre-treatment fit.

\begin{proof}[Proof of Lemma \ref{lm:w-consistency}]
For convenience, we use $\tilde{T}_0 = T_0 - h - p + 1$ to represent the effective pre-treatment sample size. Let $\hat{C}_{-1}$ denote the $\tilde{T}_0 \times N$ matrix of estimated cyclical components for the $N$ control units prior to $T_0$, and $\hat{C}_{1}$ the corresponding $\tilde{T}_0 \times 1$ vector for the treated unit.
Similarly, let $\Lambda_{-1}$ denote the $N \times L$ matrix of factor loadings for the control units, and $\Lambda_{1}$ the $1 \times L$ vector for the treated unit.
We use $F_0$ to denote the $\tilde{T}_0 \times L$ matrix of common factors before treatment.
Let $\hat{u}_{-1}$ and $\varepsilon_{-1}$ be the $\tilde{T}_0 \times N$ matrices of estimation errors and idiosyncratic shocks, respectively, for the control units; similarly, let $\hat{u}_{1}$ and $\varepsilon_{1}$ be the corresponding $\tilde{T}_0 \times 1$ vectors for the treated unit. $\hat{w}$ and $w$ are $N \times 1$ vectors of the estimated synthetic weights and true weights.

Let $w^0 = (w_2^0, w_3^0, \dots, w_{N+1}^0)'$ be a $N\times 1$ vector of infeasible weights. Given that \begin{align*}
     F_0\Lambda_{-1}'\cdot w^0  = F_0\Lambda_{1}',
\end{align*}
we have
\begin{align*}
    \hat{w} =& (\hat{C}_{-1}'\hat{C}_{-1})^{-1}(\hat{C}_{-1})'\hat{C}_{1} \\
    =& w^0 + (\hat{C}_{-1}'\hat{C}_{-1})^{-1}(\hat{C}_{-1})'(\varepsilon_1 + \hat{u}_1 - \varepsilon_{-1}\cdot w^0 - \hat{u}_{-1}\cdot w^0).
\end{align*}
We first show that $\frac{1}{\tilde{T_0}}\hat{C}_{-1}'\hat{C}_{-1}$ converges to a positive definite matrix.
Use \begin{align*}
    \hat{C}_{-1} = F_0\Lambda_{-1}' + \varepsilon_{-1} + \hat{u}_{-1}, 
\end{align*}
we have
\begin{align*}
\frac{1}{\tilde{T_0}}\hat{C}_{-1}'\hat{C}_{-1} =& \frac{1}{\tilde{T_0}} \Lambda_{-1}F_0'F_0\Lambda_{-1}' + \frac{1}{\tilde{T_0}}\Lambda_{-1}F_0' \varepsilon_{-1}  + \frac{1}{\tilde{T_0}}\Lambda_{-1}F_0 '\hat{u}_{-1} \\
&+ \frac{1}{\tilde{T_0}} \varepsilon_{-1}'F_0\Lambda_{-1}' + \frac{1}{\tilde{T_0}}\hat{u}_{-1}'F_0 \Lambda_{-1}' + \frac{1}{\tilde{T_0}}\epsilon_{-1}'\epsilon_{-1} + \frac{1}{\tilde{T_0}}\hat{u}_{-1}'\hat{u}_{-1} \\
=& \Sigma_{F_0} + \frac{1}{\tilde{T_0}}\Lambda_{-1}F_0 '\hat{u}_{-1} + \frac{1}{\tilde{T_0}}\hat{u}_{-1}'F_0 \Lambda_{-1}' + \Sigma_{\varepsilon_{-1}} + o_p(1) \\
=& \Sigma_{F_0} + \Sigma_{\varepsilon_{-1}} + o_p(1).
\end{align*}
The last convergence holds because Assumption \ref{asm:c-factor}(iii) and Assumption \ref{asm:error convergence} imply that
\begin{align*}
    \left[\frac{1}{\tilde{T_0}}\Lambda_{-1}F_0 '\hat{u}_{-1}\right]_{ij} &= \frac{1}{\tilde{T_0}} \sum_{l=1}^L \lambda_{i,l} \sum_{t= h+p}^{{T_0}}f_{l,t}\hat{u}_{jt} \\
    &\leq \sum_{l=1}^L \lambda_{i,l} \left(\frac{1}{{\tilde{T_0}}} \sum_{t= h+p}^{{T_0}}f_{l,t}^2\right)^{\frac{1}{2}} \left(\frac{1}{{\tilde{T_0}}}\sum_{t= h+p}^{{T_0}}\hat{u}_{jt}^2\right)^{\frac{1}{2}} \\
    &= o_p(1),
\end{align*}
where $[\cdot]_{ij}$ denote the $i$-th row and $j$-th column of a matrix. Next, we show that $\frac{1}{\tilde{T_0}}(\hat{C}_{-1})'(\varepsilon_1 + \hat{u}_1 - \varepsilon_{-1}\cdot w^0 - \hat{u}_{-1}\cdot w^0)$ converges in probability to a constant vector. Rewrite this expression as 
\begin{align*}
& \frac{1}{\tilde{T_0}}\left(\Lambda_{-1}'F_0 + \varepsilon_{-1}' + \hat{u}_{-1}'\right)\left(\varepsilon_1 + \hat{u}_1 - \varepsilon_{-1} w^0 - \hat{u}_{-1}w^0\right) \\
=& \frac{1}{\tilde{T_0}}\Lambda_{-1}'F_0 \left(\varepsilon_1 + \hat{u}_1 - \varepsilon_{-1}w^0 - \hat{u}_{-1}w^0\right)  + \frac{1}{\tilde{T_0}}\left( \varepsilon_{-1}' + \hat{u}_{-1}'\right)\left(\varepsilon_1 + \hat{u}_1 - \varepsilon_{-1}w^0- \hat{u}_{-1}w^0\right) \\
=& o_p(1) + \frac{1}{\tilde{T_0}}\varepsilon_{-1}'\varepsilon_{1} + \frac{1}{\tilde{T_0}}\varepsilon_{-1}'\varepsilon_{-1}w^0 + \frac{1}{\tilde{T_0}}\hat{u}_{-1}'\hat{u}_{1} + \frac{1}{\tilde{T_0}}\hat{u}_{-1}'\hat{u}_{-1}w^0 \\
& + \frac{1}{\tilde{T_0}}\varepsilon_{-1}'\hat{u}_{1} + \frac{1}{\tilde{T_0}}\varepsilon_{-1}'\hat{u}_{-1}w^0 + \frac{1}{\tilde{T_0}}\hat{u}_{-1}'\varepsilon_{1} + \frac{1}{\tilde{T_0}}\hat{u}_{-1}'\varepsilon_{-1}w^0 \\
=& o_p(1) + \Sigma_{\varepsilon_{-1}}w^0 +  \frac{1}{\tilde{T_0}}\hat{u}_{-1}'\hat{u}_{1} + \frac{1}{\tilde{T_0}}\varepsilon_{-1}'\hat{u}_{1} + \frac{1}{\tilde{T_0}}\varepsilon_{-1}'\hat{u}_{-1}w^0 + \frac{1}{\tilde{T_0}}\hat{u}_{-1}'\varepsilon_{1} + \frac{1}{\tilde{T_0}}\hat{u}_{-1}'\varepsilon_{-1}w^0. 
\end{align*}
Recall that this is a $N \times 1$ vector. We focus on the ith element of each term: 
\begin{align*}
    \left[\frac{1}{\tilde{T_0}}\hat{u}_{-1}'\hat{u}_{1}\right]_i = \frac{1}{\tilde{T_0}}\sum_{t = h+p}^{T_0} \hat{u}_{i,t}\hat{u}_{1,t} 
    \leq \left(\frac{1}{\tilde{T_0}}\sum_{t = h+p}^{T_0} \hat{u}_{i,t}^2\right)^{1/2}\left(\frac{1}{\tilde{T_0}}\sum_{t = h+p}^{T_0} \hat{u}_{1,t}^2\right)^{1/2} 
    = o_p(1),
\end{align*}
and
\begin{align*}
    \left[\frac{1}{\tilde{T_0}}\varepsilon_{-1}'\hat{u}_{1}\right]_i = \frac{1}{\tilde{T_0}}\sum_{t = h+p}^{T_0} \varepsilon_{i,t}\hat{u}_{1,t} 
    \leq \left(\frac{1}{\tilde{T_0}}\sum_{t = h+p}^{T_0} \varepsilon_{i,t}^2\right)^{1/2}\left(\frac{1}{\tilde{T_0}}\sum_{t = h+p}^{T_0} \hat{u}_{1,t}^2\right)^{1/2} 
    = o_p(1).
\end{align*}
Similarly, we can show that $\frac{1}{\tilde{T_0}}\varepsilon_{-1}'\hat{u}_{-1}w^0 + \frac{1}{\tilde{T_0}}\hat{u}_{-1}'\varepsilon_{1} + \frac{1}{\tilde{T_0}}\hat{u}_{-1}'\varepsilon_{-1}w^0 = o_p(1)$. Therefore, $$\frac{1}{\tilde{T_0}}(\hat{C}_{-1})'(\varepsilon_1 + \hat{u}_1 - \varepsilon_{-1}\cdot w^0 - \hat{u}_{-1}\cdot w^0) \overset{p}{\rightarrow} \Sigma_{\varepsilon_{-1}}w^0,$$ which concludes the proof. 
\end{proof}

\begin{proof}[Proof of Theorem \ref{thm:unbiasedness}]
    Given any weight vector $(w_{2}, w_3, \cdots, w_{N+1})$, we have 
    \begin{align}\label{eqn:basic cyclical difference}
        c_{1,t} - \sum_{i = 2}^{N+1}w_i \hat{c}_{i,t} &=  \lambda_1 f_t + \varepsilon_{1,t} - \left(\sum_{i = 2}^{N+1}w_i \lambda_i\right)f_t - \sum_{i = 2}^{N+1}w_i\varepsilon_{i,t} - \sum_{i = 2}^{N+1}w_i \hat{u}_{i,t} \nonumber \\
        &= \left( \lambda_1 - \sum_{i = 2}^{N+1}w_i \lambda_i\right)f_t + \varepsilon_{1,t} - \sum_{i = 2}^{N+1}w_i\varepsilon_{i,t} - \sum_{i = 2}^{N+1}w_i \hat{u}_{i,t}. 
    \end{align}
Let \(\hat{c}_i\), \(\varepsilon_i\), and \(\hat{u}_i\) represent \((T_0 - h - p + 1) \times 1\) vectors, where the \(t\)-th elements correspond to \(\hat{c}_{i,t}\), \(\varepsilon_{i,t}\), and \(\hat{u}_{i,t}\), respectively. Define 
$\phi_t = f_t'(F_0'F_0)^{-1}F_0'$ and we have
\begin{align*}
    \phi_t \left(F_0 \lambda_1' + \varepsilon_{1}  - \sum_{i = 2}^{N+1}w_i \hat{c}_{i}\right) &= \phi_t\left[F_0\left( \lambda_1' - \sum_{i = 2}^{N+1}w_i \lambda_i'\right) + \varepsilon_{1} - \sum_{i = 2}^{N+1}w_i\varepsilon_{i} - \sum_{i = 2}^{N+1}w_i \hat{u}_{i}\right] \\
    &= f_t'\left( \lambda_1' - \sum_{i = 2}^{N+1}w_i \lambda_i'\right) + \phi_t \varepsilon_{1}- \phi_t\sum_{i = 2}^{N+1}w_i\varepsilon_{i} \ - \phi_t\sum_{i = 2}^{N+1}w_i \hat{u}_{i}.
\end{align*}
   Rearranging terms yields 
   \begin{align*}
       \left( \lambda_1 - \sum_{i = 2}^{N+1}w_i \lambda_i\right)f_t = \phi_t \left(F_0 \lambda_1'  - \sum_{i = 2}^{N+1}w_i \hat{c}_{i}\right)  + \phi_t\sum_{i = 2}^{N+1}w_i\varepsilon_{i} \ + \phi_t\sum_{i = 2}^{N+1}w_i \hat{u}_{i},
   \end{align*}
   which is then substituted into equation (\ref{eqn:basic cyclical difference}) to obtain the following:
   \begin{align*}
       c_{1,t} - \sum_{i = 2}^{N+1}w_i \hat{c}_{i,t} =&\phi_t \left(F_0 \lambda_1'  - \sum_{i = 2}^{N+1}w_i \hat{c}_{i}\right)   + \phi_t\sum_{i = 2}^{N+1}w_i\varepsilon_{i} \ + \phi_t\sum_{i = 2}^{N+1}w_i \hat{u}_{i} + \varepsilon_{1,t} - \sum_{i = 2}^{N+1}w_i\varepsilon_{i,t} - \sum_{i = 2}^{N+1}w_i \hat{u}_{i,t} \\
       =& \phi_t\left(F_0 \lambda_1'  - F_0\sum_{i = 2}^{N+1}w_i \lambda_i' \right) + \varepsilon_{1,t} - \sum_{i = 2}^{N+1}w_i\varepsilon_{i,t} - \sum_{i = 2}^{N+1}w_i \hat{u}_{i,t}.
   \end{align*}
   For convenience, we follow \cite{abadie2010synthetic} in assuming a perfect pre-treatment fit, that is, the minimization problem in (\ref{eqn:cycle-weights}) attains zero, and $ \hat{c}_{1} = \sum_{i=2}^{N+1}\hat{w}_i\hat{c}_{i}$,
    so
    \begin{align*}
        F_0 \lambda_1'  - F_0\sum_{i = 2}^{N+1}w_i \lambda_i' = \sum_{i = 2}^{N+1}w_i(\varepsilon_{i}+\hat{u}_i)-\varepsilon_1 - \hat{u}_1.
    \end{align*}
Therefore, 
\begin{align*}
    c_{1,t} - \sum_{i = 2}^{N+1}\hat{w}_i \hat{c}_{i,t} &= \phi_t\sum_{i = 2}^{N+1}\hat{w}_i\hat{u}_i - \phi_t\hat{u}_1 - \sum_{i = 2}^{N+1}\hat{w}_i \hat{u}_{i,t} + \phi_t\sum_{i = 2}^{N+1}\hat{w}_i\varepsilon_i- \phi_t\varepsilon_1   - \sum_{i = 2}^{N+1}\hat{w}_i(\varepsilon_{i,t}-\varepsilon_{1,t})  \\
    &\equiv R_{1t} + R_{2t} + R_{3t} + R_{4t} + R_{5t} + R_{6t}.
\end{align*}
Using Lemma \ref{lm:w-consistency}, we rewrite $R_{1t}$ as
\begin{align*}
    R_{1t} &= \sum_{i = 1}^{N+1} \hat{w}_i \left[\sum_{s=h+p}^{T_0} f_t'\left(F_0'F_0\right)^{-1}f_s \hat{u}_{i,s}\right]\\
    &= \sum_{i = 1}^{N+1} \left(w_i^0 + c_i + o_p(1)\right)\left[\sum_{s=h+p}^{T_0} f_t'\left(F_0'F_0\right)^{-1}f_s \hat{u}_{i,s}\right].
\end{align*}
Following Appendix B of \cite{abadie2010synthetic}, we can apply the Cauchy-Schwarz inequality and Assumption \ref{asm:c-factor} to obtain the following bound:
\begin{align} \label{ineq:cauchy of phi}
    \left(f_t' \left(F_0'F_0\right)^{-1}f_s \right)^2 &\leq \left(f_t' \left(F_0'F_0\right)^{-1}f_t \right)\left(f_s' \left(F_0'F_0\right)^{-1}f_s \right) \leq \left(\frac{\overline{F}^2 L}{\left(T_0 - h-p+1\right)\underline{\xi}}\right)^2.
\end{align}
Using Assumption \ref{asm:error convergence} and applying the Cauchy-Schwarz inequality to $R_{1t}$, we have
    \begin{align*}
        R_{1t}^2 & \leq \frac{\left(\overline{F}^2 L/\underline{\xi}\right)^2}{T_0 - h-p+1}\left(\sum_{i = 1}^{N+1} \left(w_i^0 + c_i + o_p(1)\right)^2 \right)  \left(\sum_{i = 1}^{N+1}\sum_{s=h+p}^{T_0} \hat{u}_{i,s} ^2 \right) = o_p\left(T_0\right)/T_0 = o_p(1),
    \end{align*}
which implies that $R_{1t} = o_p\left(1\right)$.
Similarly, we can show that $R_{2t} = o_p\left(1\right)$. With Lemma \ref{lm:w-consistency} and the convergence of $\hat{u}_{i,t}$ by Assumption \ref{asm:error convergence}, 
\begin{align*}
        R_{3t} = \sum_{i = 1}^{N+1}\left(w_i^0 + c_i + o_p(1)\right)\hat{u}_{i,t} = o_p(1).
    \end{align*}
Note that if  the cyclical component follows an exact low-dimensional
factor structure $c_{i,t} = \lambda_i'f_t$, then 
$$c_{1,t} - \sum_{i = 2}^{N+1} \hat{w}_i\hat{c}_{i,t} = R_{1t}+R_{2t}+R_{3t} = o_p(1),$$
which proves part (ii) of Theorem \ref{thm:unbiasedness}. 

We then study $ R_{4t}+R_{5t}+R_{6t}$ to complete the proof of part (i) of Theorem \ref{thm:unbiasedness}. 
The upper bound of the bias of $\vert R_{4t} \vert$ was analyzed in \cite{abadie2010synthetic}. We can rewrite $R_{4t}$ as
\begin{align*}
    R_{4t} & = \sum_{i = 2}^{N+1} \hat{w}_i \sum_{s = h+p}^{T_o} f_t'\left(F_0'F_0\right)^{-1}f_s \varepsilon_{i,s} \equiv \sum_{i = 2}^{N+1} \hat{w}_i \overline{\varepsilon}_{i,t},
\end{align*}
where $\overline{\varepsilon}_{i,t} = f_t'\left(F_0'F_0\right)^{-1}\sum_{s = h+p}^{T_o} f_s \varepsilon_{i,s} $.
Using the Cauchy-Schwarz inequality and Assumption \ref{asm:c-factor},
\begin{align*}
    \mathbb{E}\left[\sum_{i = 2}^{N+1} \vert {w}_i^0\overline{\varepsilon}_{i,t}\vert\right] \leq \overline{w}\sum_{i = 2}^{N+1}\left(\mathbb{E}\left[ \overline{\varepsilon}_{i,t}^2\right] \right)^{1/2}, 
\end{align*}
where $\overline{w}^0 \geq\vert w_i^0 \vert $ is some constant bound for $w^0$. 
Then, using the Rosenthal’s inequality to the second term above, we obtain
\begin{align*}
    \mathbb{E}[\overline{\varepsilon}_{i,t}^2] &\leq 
    \bar{F}^4\mathbb{E}\left[ \left(\sum_{s = h+p}^{T_0}\mathbf{1}'\left(F_0'F_0\right)^{-1}\mathbf{1}\varepsilon_{i,s} \right)^2\right] \\
    &\leq \left(\frac{\overline{F}^2 L}{\left(T_0 - h-p+1\right)\underline{\xi}}\right)^2 \mathbb{E}\left[\left(\sum_{s = h+p}^{T_0} \varepsilon_{i,s}  \right)^2\right]\\
    & \leq C\left(\frac{\overline{F}^2 L}{\left(T_0 - h-p+1\right)\underline{\xi}}\right)^2 \left(\sum_{s=h+p}^{T_0} \mathbb{E}[\varepsilon_{i,s}^2]\right),
\end{align*}
where $C$ is a constant independent of $T_0$, see \cite{ibragimov2002exact}. 
Collecting the above steps, we have verified that the following absolute moment goes to zero:
\begin{align*}
   \mathbb{E}\left[\sum_{i = 2}^{N+1} \vert {w}_i^0\overline{\varepsilon}_{i,t}\vert\right] \leq  \overline{w}^0\left(\mathbb{E}\left[\sum_{i = 2}^{N+1}  \overline{\varepsilon}_i^2\right] \right)^{1/2} \leq& \sqrt{C}\frac{\overline{F}^2 L}{\left(T_0 - h-p+1\right)\underline{\xi}}\sqrt{\sum_{i = 2}^{N+1}\sum_{s=h+p}^{T_0} \mathbb{E}[\varepsilon_{i,s}^2]}\\
   =& O\left(\frac{1}{T_0 - h-p+1}\right).
\end{align*}
By the Markov inequality,  $R_{4t} \overset{p}{\rightarrow} 0$.
The consistency of $R_{5t}$ is similar except that it does not involve $\hat{w}_i$'s. Hence,
\begin{align*}
    \hat{Y}_{1,t}(0) - Y_{1,t}(0) = o_p(1) + R_{6t}. 
\end{align*}

Finally, because $R_{6t}$ is a weighted average of idiosyncratic errors in the post-treatment period $t$, $\varepsilon_{i,t}$ are independent of the weights $\hat{w}_{i}$ determined by the pre-treatment potential outcomes. Therefore, we can conclude that $\mathbb{E}[R_{6t}] = 0$, and $\hat{Y}_{1,t}(0)$ is asymptotically unbiased. 
\end{proof}

\begin{proof} [Proof of Lemma \ref{lm:c-stationary}]
    See Proposition 4 of \cite{hamilton2018you}. 
\end{proof}

\begin{proof} [Proof of Lemma \ref{lm:hamilton-filter-consistency}]
    By Proposition 4 of \cite{hamilton2018you}, we can obtain the consistency of $(\hat{\alpha}_{i,0},\cdots,\hat{\alpha}_{i,p})$, which implies the consistency of  $(\hat{\tau}_{i,t},\hat{c}_{i,t})$ for all $1 \leq i \leq N+1, h+p \leq t \leq T_0$. For a post-treatment $T_0+1 \leq t \leq T_0 + h$, we have
    \begin{align*}
        \hat{\tau}_{1,t} & = \hat{\alpha}_{1,0} + \hat{\alpha}_{1,1} Y_{1,t-h-1}(0) + \cdots + \hat{\alpha}_{1,p} Y_{1,t-h-p+1}(0) \\
        & \overset{p}{\rightarrow} \alpha_{1,0} + \alpha_{1,1} Y_{1,t-h-1}(0) + \cdots + \alpha_{1,p} Y_{1,t-h-p+1}(0) = \tau_{1,t}
    \end{align*}
    because of the Slutsky theorem. Similarly, we have $\hat{u}_{i,t} = o_p(1)$ for any $t \leq T_0 + h$ and $i=2,\cdots,N+1$ for all $i$ and $t$ under consideration.
\end{proof}

\section{Additional Empirical Results} \label{sec:additional-empirical}

We present additional results to supplement the analysis in Section \ref{sec:HongKong}, conducting the exercise with the donor pool used by \citet{hsiao2012panel}: mainland China, Indonesia, Japan, Korea, Malaysia, the Philippines, Singapore, Thailand, and the United States. 

Figure \ref{fig:HK_robustness_gdp_trend_cyc} highlights two features of this pool. First, this set includes economies whose per‑capita GDP is substantially lower than that of Hong Kong. Second, several donors exhibit sharp GDP movements around 1997, suggesting they were themselves affected by events coinciding with Hong Kong’s return and could therefore confound the counterfactual. Figure \ref{fig:hsiao-treatment-effect}(b) is also indicative of the latter concern, as the synthetic trajectory almost overlays Hong Kong’s actual post‑return path. The placebo tests in Figure \ref{fig:hsiao-placebo} deliver qualitatively similar patterns to those in Section \ref{sec:HongKong}, with the synthetic business cycle estimator outperforming the conventional one.

\begin{figure}[htbp!]
  \centering
  \begin{subfigure}{0.5\textwidth}
    \centering
\includegraphics[width=\textwidth]{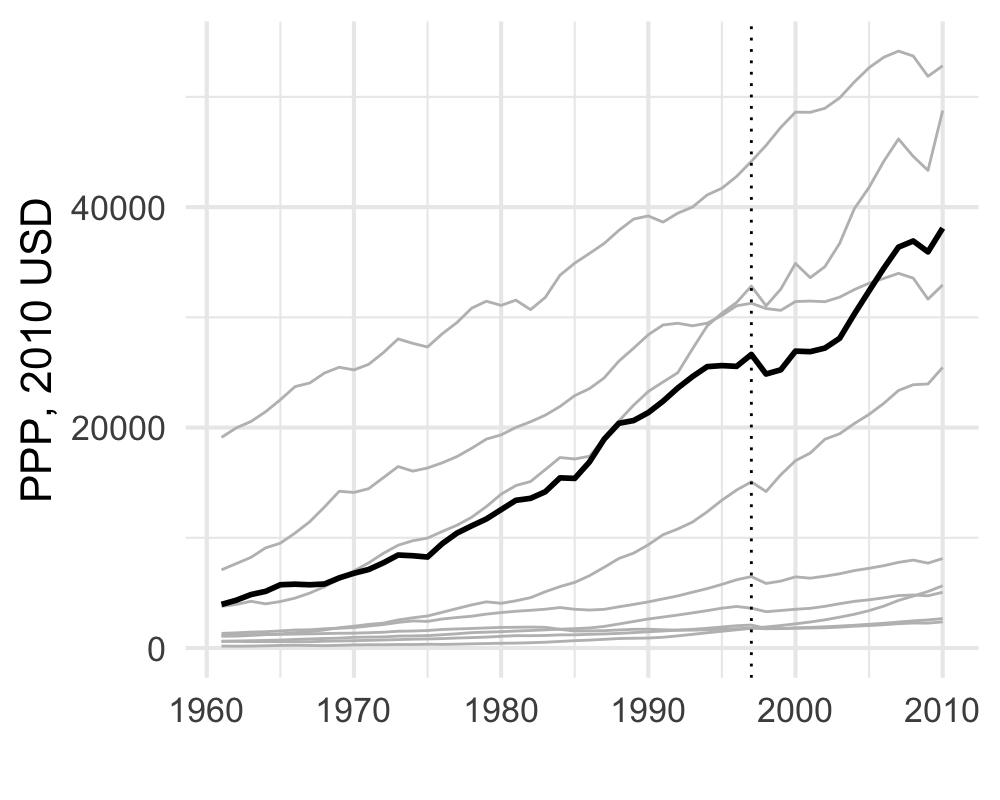}
    \caption{Raw per capita GDP}
  \end{subfigure}


  \begin{subfigure}{0.49\textwidth}
    \centering
 \includegraphics[width=\textwidth]{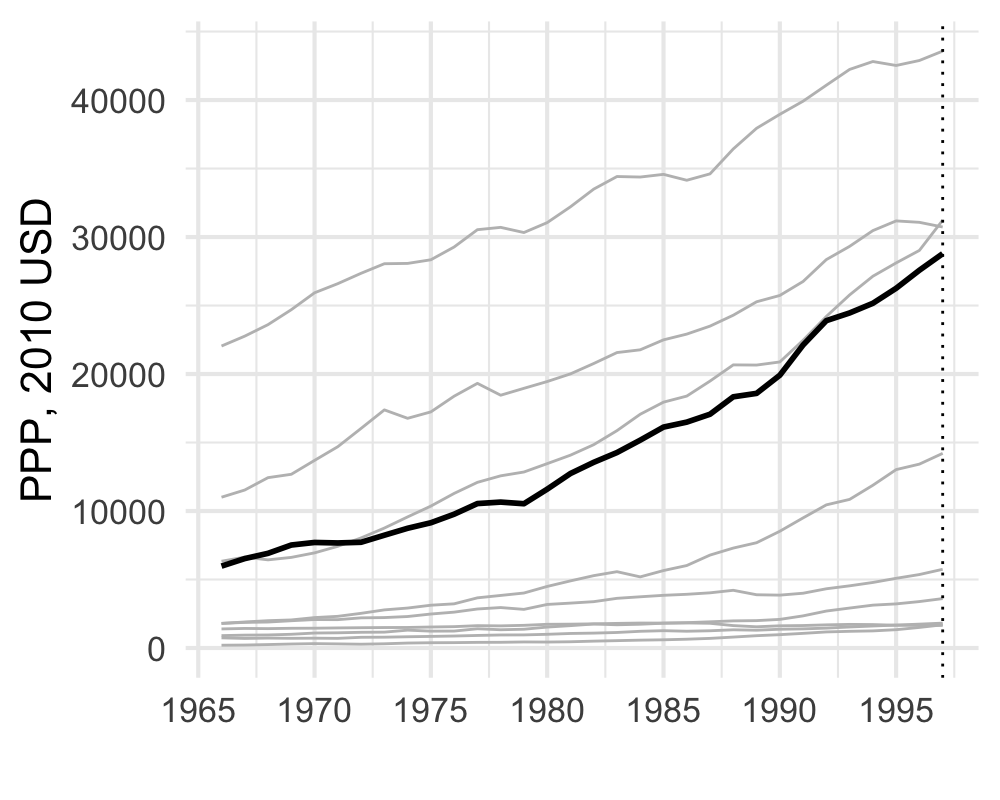}
    \caption{Trend components}
  \end{subfigure}
  \hfill
  \begin{subfigure}{0.49\textwidth}
    \centering
\includegraphics[width=\textwidth]{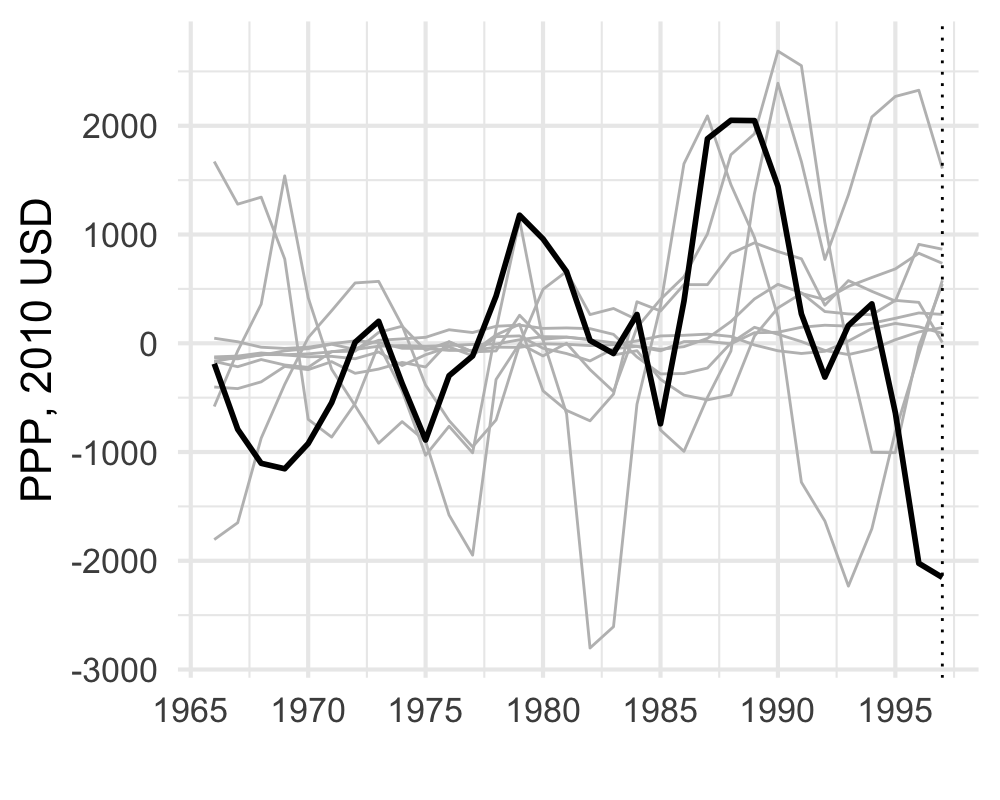}
    \caption{Cyclical components}
  \end{subfigure}
  \caption{GDP path decomposition for Hong Kong and 9 donor economies}
  \label{fig:HK_robustness_gdp_trend_cyc}
  \caption*{\footnotesize The thick black curve represents Hong Kong's GDP per capita, while the grey curves correspond to the 9 donor economies analyzed in \cite{hsiao2012panel}: mainland China, Indonesia, Japan, Korea, Malaysia, Philippines, Singapore, Thailand, and the US. Panel (a) shows the raw data; panels (b) and (c) display the trend and cyclical components, respectively, extracted with the Hamilton filter applied to the pre-unification data, using horizon $h=4$ years and $p=2$ lags.
}
\end{figure}

\begin{figure}[htbp!]
  \centering
  \begin{subfigure}{0.49\textwidth}
    \includegraphics[width=\textwidth]{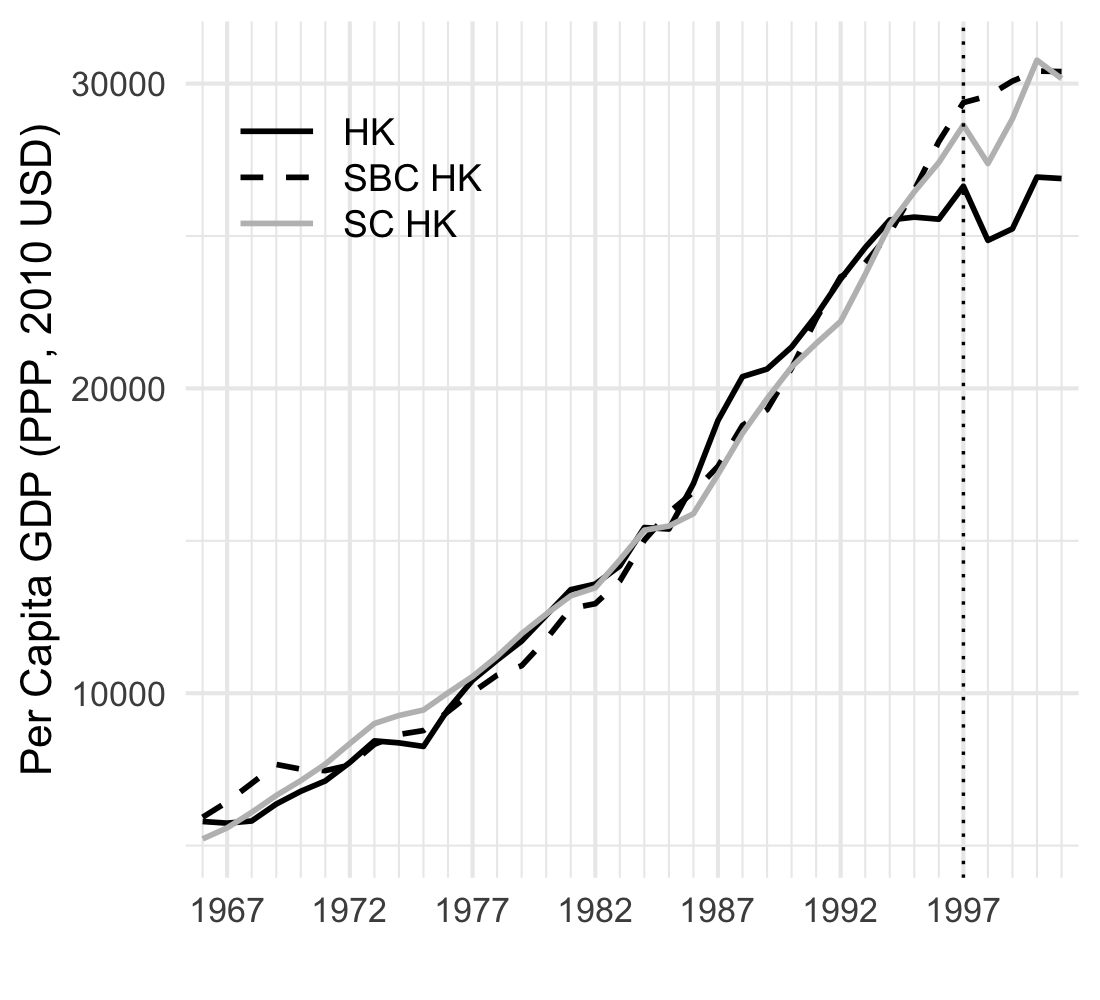}
    \caption{Non-negative weights}
  \end{subfigure}
  \hfill
  \begin{subfigure}{0.49\textwidth}
    \includegraphics[width=\textwidth]{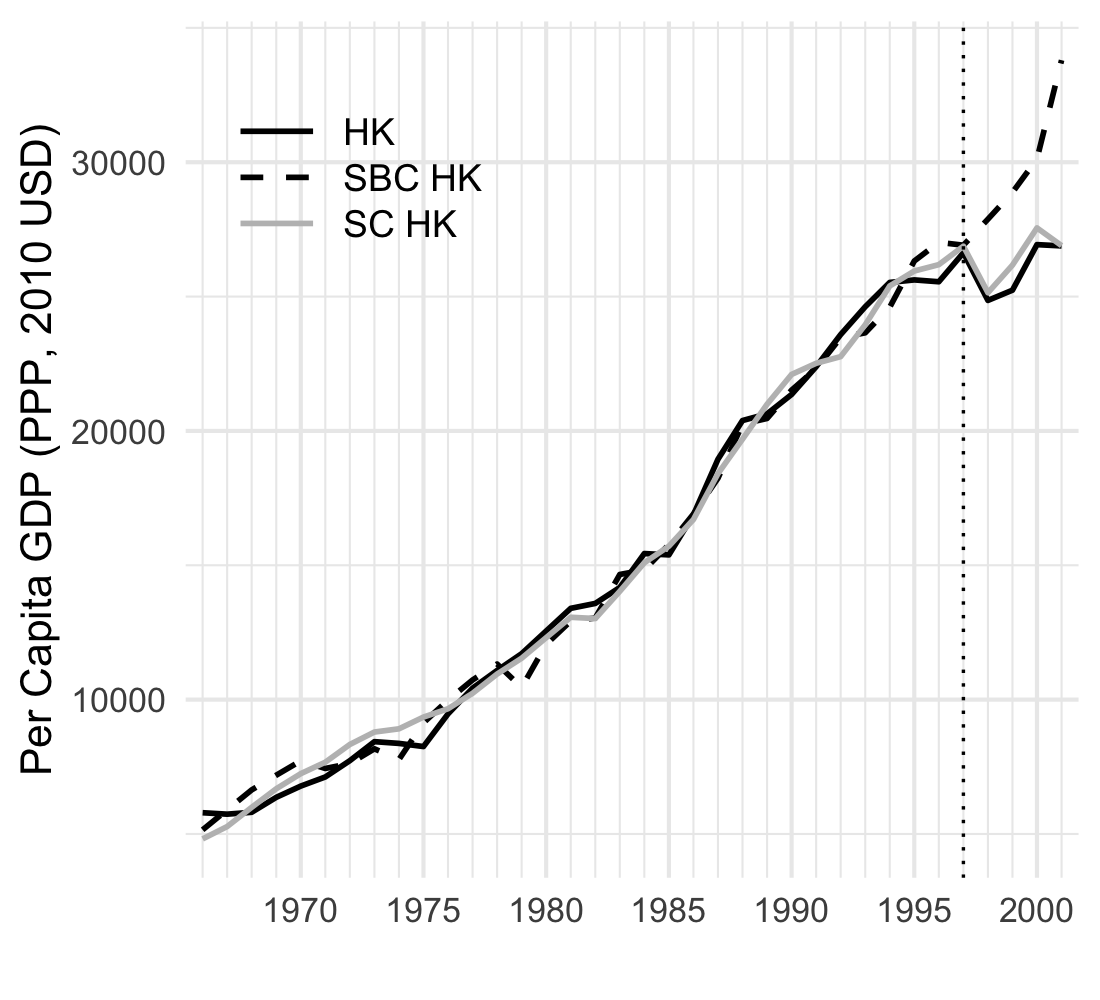}
    \caption{Signed weighted}
  \end{subfigure}
  \caption{Robustness check: treatment effect estimates of return of Hong Kong}
  \label{fig:hsiao-treatment-effect}
  \caption*{\footnotesize Hong Kong's synthetic GDP estimated using both the synthetic business cycle estimator and the conventional synthetic control estimator, compared with the actual GDP. Panels (a) and (b) impose and relax the non-negativity constraint on the weights, respectively. The estimates are obtained using the donor economies in Figure \ref{fig:HK_robustness_gdp_trend_cyc}.}
\end{figure}

\begin{figure}[htbp!]
  \centering
  \begin{subfigure}{0.49\textwidth}
    \includegraphics[width=\textwidth]{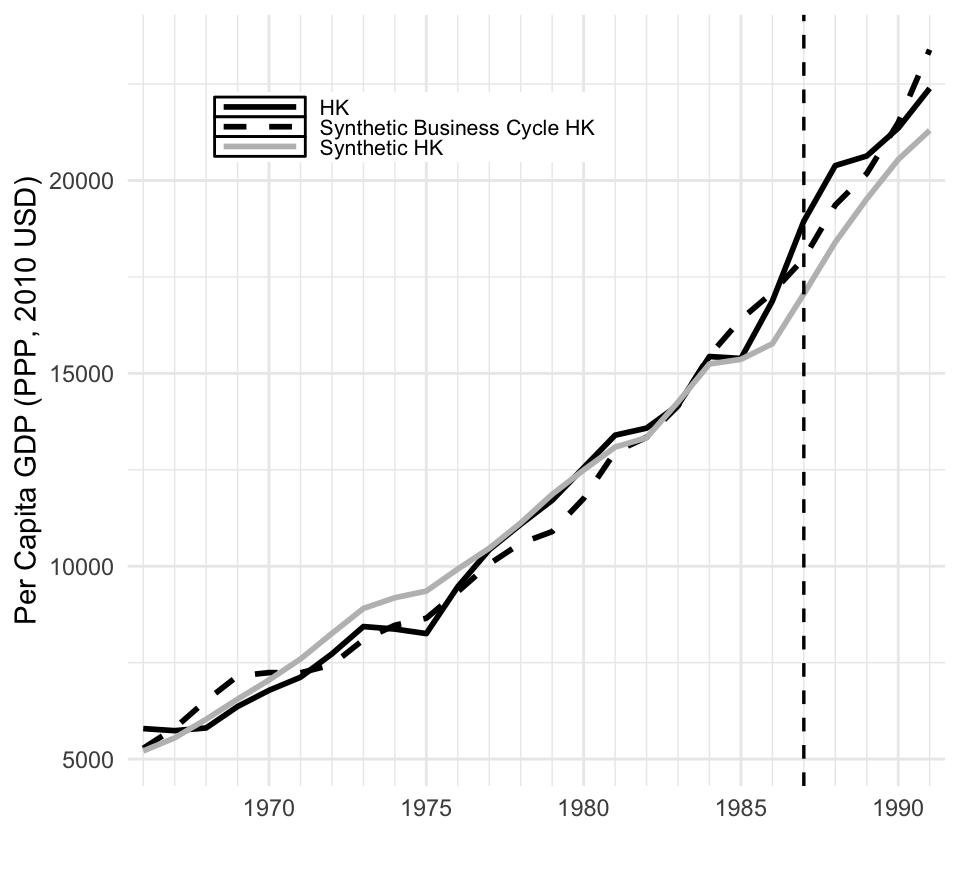}
    \caption{Non-negative weights}
  \end{subfigure}
  \hfill
  \begin{subfigure}{0.49\textwidth}
    \includegraphics[width=\textwidth]{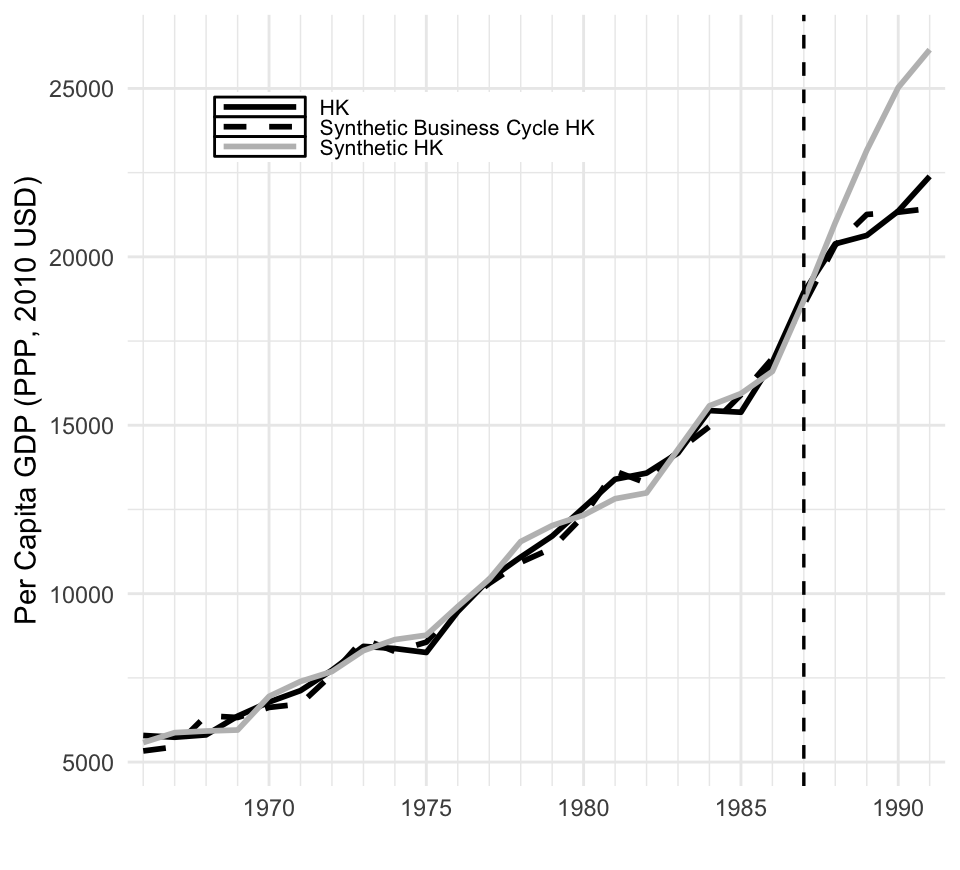}
    \caption{Signed weights}
  \end{subfigure}
  \caption{Robustness check: placebo return of Hong Kong in 1987}
  \label{fig:hsiao-placebo}
  \caption*{\footnotesize This figure applies the synthetic business cycle and the conventional synthetic control method to the placebo return of Hong Kong dated 1987. Panels (a) and (b) impose and relax the non-negativity constraint on the weights, respectively. The estimates are obtained using the donor economies in Figure \ref{fig:HK_robustness_gdp_trend_cyc}.}
\end{figure}

\end{document}